\documentclass[11pt]{article}
\usepackage[latin1]{inputenc}
\usepackage[T1]{fontenc}
\usepackage{lmodern}
\usepackage[a4paper]{geometry}
\usepackage{amsmath, amssymb}
\usepackage{amsthm}
\usepackage{dsfont}
\usepackage[english]{babel}
\usepackage{a4wide}
\usepackage{graphicx,color}
\usepackage{amsopn}
\usepackage{array}
\usepackage[nottoc, notlof, notlot]{tocbibind}
\usepackage{subfigure,enumitem}
\usepackage{mathtools}
\usepackage[colorlinks]{hyperref}
\hypersetup{
linkcolor=blue,
citecolor=blue,
}

\newcommand{\thefont}[2]{\fontsize{#1}{#2}\fontshape{n}\selectfont}
\newcommand{\1}{\rlap{\thefont{10pt}{12pt}1}\kern.16em\rlap{\thefont{11pt}{13.2pt}1}\kern.4em}

\title{Data-driven regularization of Wasserstein barycenters with an application to multivariate density registration}
\author{J\'er\'emie Bigot\footnote{J. Bigot is a member of Institut Universitaire de France.}, Elsa Cazelles \& Nicolas Papadakis  \\
\\  Institut de Math\'ematiques de Bordeaux et CNRS  (UMR 5251)   \\ Universit\'e de Bordeaux }

\begin{document}

\theoremstyle{plain}
\newtheorem{thm}{Theorem}[section]
\theoremstyle{plain}
\newtheorem{prop}{Properties}[section]
\theoremstyle{plain}
\newtheorem{hyp}{Assumption}[section]
\theoremstyle{plain}
\newtheorem{proposition}{Proposition}[section]
\theoremstyle{definition}
\newtheorem{defi}[thm]{Definition}
\theoremstyle{plain}
\newtheorem{lemma}[thm]{Lemma}
\newtheorem{cor}[thm]{Corollary}
\newtheorem{rmq}{Remark}[section]
\theoremstyle{definition}
\newtheorem{ex}{Example}[section]
\newtheorem{appli}{Application}[section]

\newcommand{\E}{{\mathbb E}}
\newcommand{\R}{{\mathbb R}}
\newcommand{\N}{{\mathbb N}}
\newcommand{\Z}{{\mathbb Z}}
\renewcommand{\P}{{\mathbb P}}
\newcommand{\G}{{\mathbb G}}
\renewcommand{\L}{\mathbb L}

\newcommand{\PP}{{\mathcal P}}
\newcommand{\HH}{{\mathcal H}}
\newcommand{\DD}{{\mathcal D}}
\newcommand{\BB}{{\mathcal B}}
\newcommand{\MM}{{\mathcal M}}

\newcommand{\bX}{\boldsymbol{X}}
\newcommand{\bY}{\boldsymbol{Y}}
\newcommand{\bnu}{\boldsymbol{\nu}}
\newcommand{\bmu}{\boldsymbol{\mu}}
\newcommand{\boldeta}{\boldsymbol{\eta}}
\newcommand{\bsigma}{\boldsymbol{\sigma}}
\newcommand{\balpha}{\boldsymbol{\alpha}}
\newcommand{\bfun}{\boldsymbol{f}}
\newcommand{\bm}{\boldsymbol{m}}
\newcommand{\bGamma}{\boldsymbol{\Gamma}}
\newcommand{\br}{\boldsymbol{r}}
\newcommand{\bq}{\boldsymbol{q}}
\newcommand{\bG}{\boldsymbol{G}}

\newcommand{\tr}{\tilde{r}}
\newcommand{\tH}{\tilde{H}}
\newcommand{\tg}{\tilde{g}}
\newcommand{\tS}{\tilde{\Sigma}}

\newcommand{\argmin}{argmin}
\newcommand{\diag}{\operatorname{diag}}

\def\argmin{\mathop{\rm arg \; min}\limits}%
\newcommand{\uargmin}[1]{\underset{#1}{\argmin}\;}

\numberwithin{equation}{section} 

\maketitle

\thispagestyle{empty}

\begin{abstract}
We present a framework to simultaneously align and smooth data in the form of  multiple point clouds sampled from unknown densities with support in a $d$-dimensional Euclidean space. This work is motivated by applications in bioinformatics where researchers aim to  automatically homogenize 
large datasets to compare and analyze characteristics within a same cell population. Inconveniently, the information acquired is most certainly noisy
due to mis-alignment caused by technical variations of the environment. To overcome this problem, we propose to register multiple point clouds by using the notion of regularized barycenters (or Fr\'{e}chet mean) of a set of probability measures with respect to the Wasserstein metric.
A first approach consists in penalizing a Wasserstein barycenter with a convex functional as recently proposed in \cite{BCP17}. 
 A second strategy is to transform the Wasserstein metric itself into an entropy regularized transportation cost between probability measures as introduced in \cite{cuturi2013sinkhorn}. The main contribution of this work is to propose data-driven choices for the regularization parameters involved in each approach using the Goldenshluger-Lepski's principle. Simulated data sampled from Gaussian mixtures are used to illustrate each method, and an application to the analysis of flow cytometry data is finally proposed. 
 This way of choosing of the regularization parameter for the Sinkhorn barycenter is also analyzed through the prism of an oracle inequality that relates the error made by such data-driven estimators to the one of an ideal estimator.
\end{abstract}

\section{Introduction}

\subsection{Motivations}

This paper is concerned with the  problem of aligning (or registering) elements of a dataset that can be modeled  as $n$ random densities, or more generally, probability measures supported on $\R^{d}$.  As raw data in the form of densities are generally not directly available, we focus on the setting where one has access to a set of  random vectors  $(X_{i,j})_{1 \leq j \leq p_{i}; \;  1 \leq i \leq n}$ in $\R^{d}$  organized in the form of $n$ subjects (or multiple point clouds), such that $X_{i,1},\ldots,X_{i,p_{i}}$ are iid observations sampled from a random density  $\bfun_{i}$  for each $1 \leq i \leq n$. In  presence of phase variation in the observations due to mis-alignment in the acquisition process, it is necessary to use a registration step to obtain meaningful notions of mean and variance from the analysis of the dataset. In Figure \ref{fig:ex_intro}(a),  we display a simulated example of  $n=2$ random distributions made of observations sampled from random Gaussian mixtures $\bfun_{i}$. 
Certainly, one can estimate a mean density using a preliminary smoothing step (with a kernel $K$ and data-driven choices of the bandwidth parameters $(h_i)_{i=1,\ldots,n}$) followed by standard averaging, that is considering
\begin{equation}
\label{eq:Eucli_mean}
\bar{f}_{n,p}(x) = \frac{1}{n} \sum_{i=1}^{n} \frac{1}{p_{i} h_{i}} \sum_{j=1}^{p_{i}} K \left( \frac{x - X_{i,j}}{h_{i}} \right), \; x \in \R^d.
\end{equation}
Unfortunately this leads to an estimator which is not consistent with the shape of the $\bfun_{i}$'s. Indeed, the estimator $\bar{f}_{n,p}$ (Euclidean mean) has four modes due to mis-alignment of the data from different subjects.

\begin{figure}[h]
\centering
\subfigure[]{\includegraphics[width=0.45 \textwidth,height=0.45\textwidth]{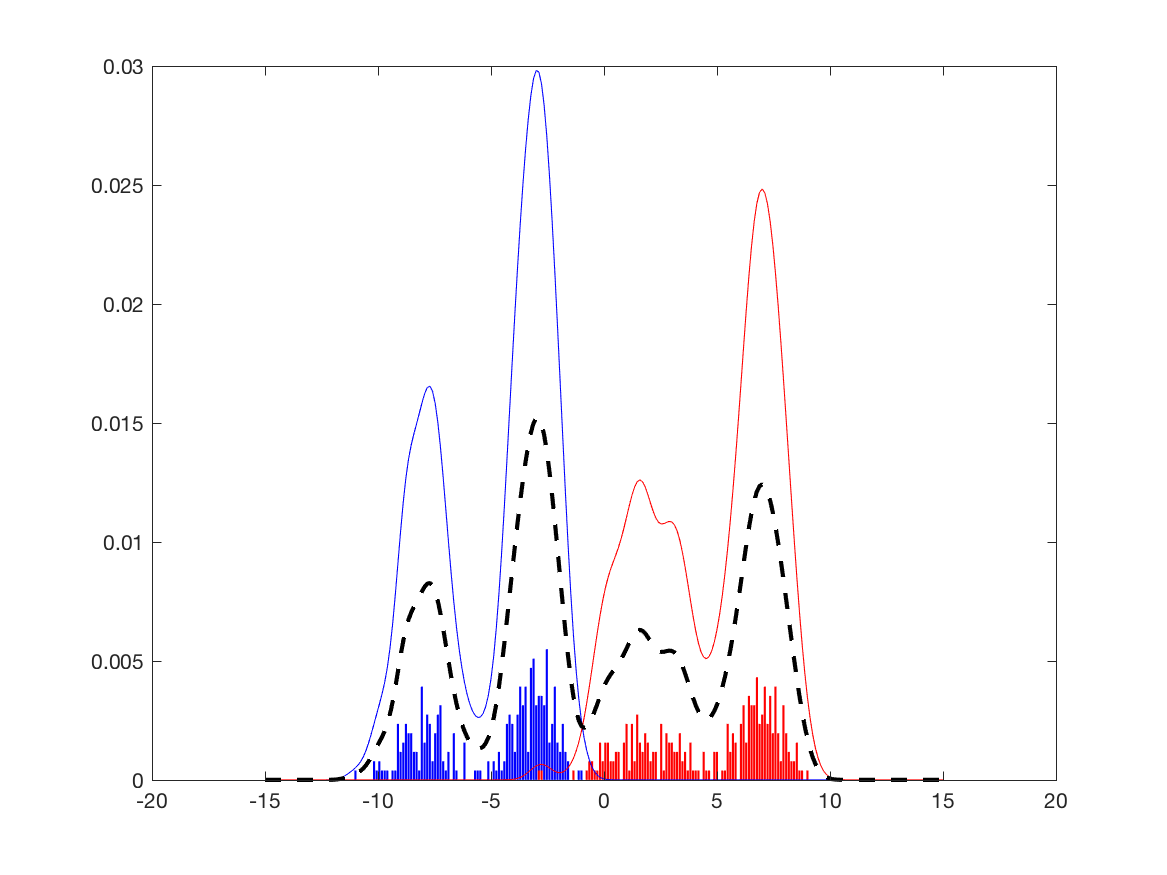}}
\subfigure[]{\includegraphics[width=0.45 \textwidth,height=0.45\textwidth]{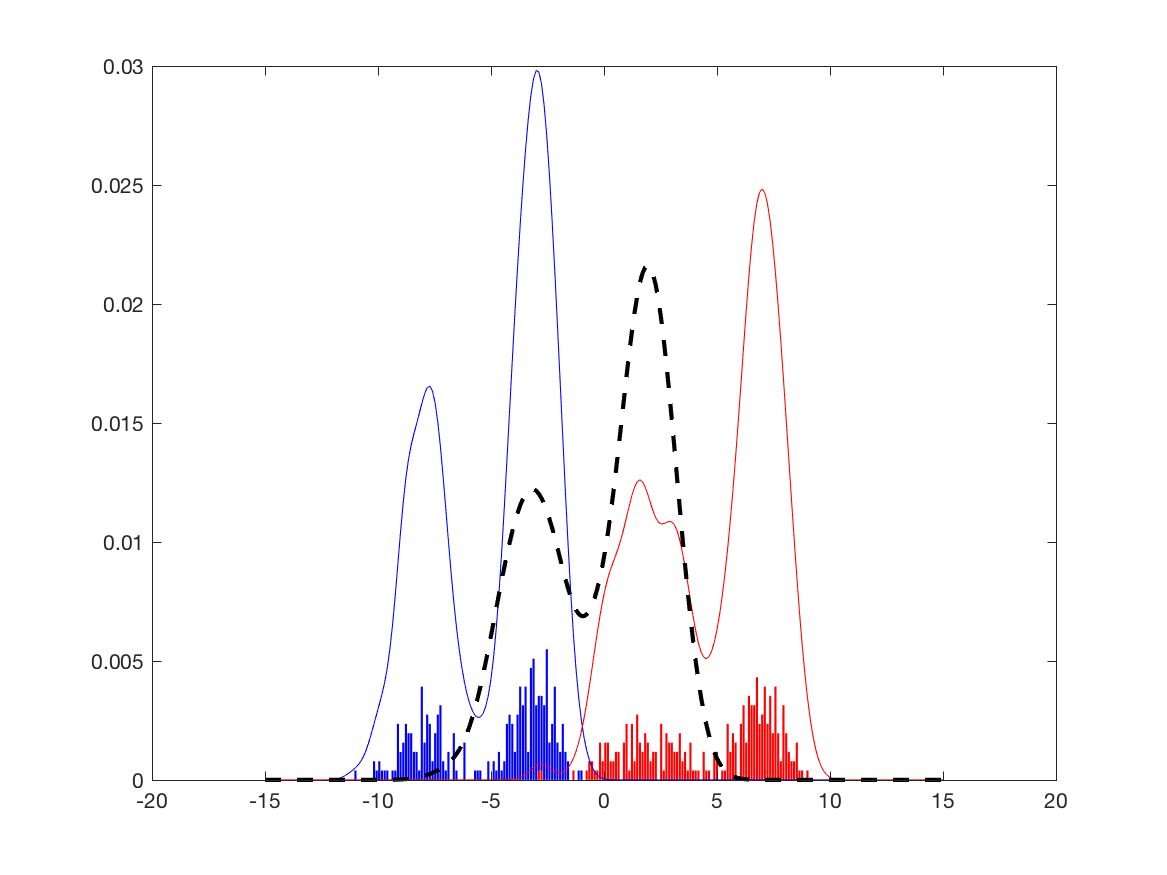}}
\caption{A simulated example of  $n=2$ subjects made of $p_1 = p_2 = 300$ observations sampled from Gaussian mixtures with random means and variances. The red and blue bar graphs are histograms with bins of equal and very small size to display the two sets of observations.  The red and blue curves represent the kernel density estimators associated to each subject with data-driven choices (using cross-validation) of the bandwidths. (a) The dashed black curve is the Euclidean mean $\bar{f}_{n,p}$ of the red and blue densities. (b) The solid black curve is the entropy regularized Wasserstein barycenter $\hat{\br}^{\hat{\varepsilon}}_{n,p}$ (defined in \eqref{eq:defr}) of the raw data using a Sinkhorn divergence and the numerical approach from \cite{CuturiPeyre}, with a data-driven choice for $\hat{\varepsilon}=1.6$.}
\label{fig:ex_intro}
\end{figure}

The need to account for phase variability in the statistical analysis of such datasets is a well-known problem in various scientific fields. For the one-dimensional case ($d=1$), examples can be found in biodemographic and genomics studies \cite{ZhangMuller},  economics  \cite{MR1946423}, and in the analysis of spike trains in neuroscience  \cite{Srivastava} or functional connectivity between brain regions \cite{PetersenMuller}. For $d \geq 2$ the issue of data registration  arises in the statistical analysis of spatial point processes  \cite{Gervini2016474,Pana17} or flow cytometry data \cite{hahne2010per,pyne2014joint}.

\subsection{Related works}

In this work, in order to simultaneously align and smooth multiple point clouds (in the idea of recovering the underlying density function), we average the data using the notion of Wasserstein barycenter (as introduced in the seminal work \cite{agueh2011barycenters}). Surely, this barycenter has been shown to be a relevant tool to account for phase variability in density registration \cite{BGKL18,Pana15,Pana17}. A Wasserstein barycenter is a Fr\'echet mean \cite{fre} in the space $\PP_2(\Omega)$ of probability measures with finite second moment supported on a convex domain $\Omega \subset \R^{d}$. It is endowed with the Wasserstein metric $W_2$ defined as
$$
W_2(\mu,\nu)=\inf_{\pi\in\Gamma(\mu,\nu)}\left(\iint_{\Omega^2}\vert x-y\vert^2d\pi(x,y)\right)^{1/2}, \quad \mbox{ for } \mu,\nu \in \PP_2(\Omega),
$$
where $\Gamma(\mu,\nu)$ is the set of probability measures on the product space $\Omega\times\Omega$ with respective marginals $\mu$ and $\nu$, and $\vert \cdot \vert$ denotes the usual Euclidean norm on $\R^d$.

In the case of a finite space $\Omega_{N}=\{x_1,\ldots,x_N\}\in(\R^d)^N$ of cardinal $N$, a discrete probability distribution $r$ (with fixed support included in $\Omega_{N}$) is identified by a vector in the simplex $\Sigma_N=\{r=(r_1,\ldots,r_N)\in\R^N_{+} \ \mbox{with} \ \sum_{k=1}^Nr_k=1\}$ such that $r=\sum_{k=1}^Nr_k\delta_{x_k}$ where $\delta_x$ is the Dirac distribution at $x$. The Wasserstein distance between two discrete distributions $r$ and $q$ in $\Sigma_N$ then becomes
$$W_2(r,q):=\underset{U\in U(r,q)}{\min}\  \langle C, U\rangle^{1/2},$$
where  the set of couplings is defined as $U(r,q):=\{U\in\R^{N\times N}_{+} \ \mbox{such that} \ U\mathds{1}_N=r,\ U^T\mathds{1}_N=q\}$ with $\mathds{1}_N$ the $N$ dimensional vector with all entries equal to $1$ and $C$ the cost matrix given by $C_{ml}=\vert x_m-x_l\vert^2$, for all $m,l\in\{1,\ldots,N\}$.

In what follows, we consider  two approaches for the computation of a regularized Wasserstein barycenter of $n$ discrete probability measures given by
\begin{equation}
\hat{\bnu}_i^{p_i}=\frac{1}{p_i}\sum_{j=1}^{p_i}\delta_{X_{i,j}} \qquad \mbox{for} \quad 1 \leq i \leq n, \label{eq:data}
\end{equation}
from observations $(X_{i,j})_{1 \leq j \leq p_{i}; \;  1 \leq i \leq n}$.

\subsubsection{Penalized Wasserstein barycenters} 

Adding a convex penalization term to the definition of an empirical Wasserstein barycenter  \cite{agueh2011barycenters}  leads to the estimator
\begin{equation}
\hat{\bmu}^{\gamma}_{n,p} = \uargmin{\mu\in\PP_2(\Omega)}  \frac{1}{n}\sum_{i=1}^nW_2^2(\mu,\hat{\bnu}_i^{p_i})+\gamma E(\mu), \label{eq:defmu}
\end{equation}
where $\gamma > 0$ is a regularization parameter, and  $E : \PP_{2}(\Omega) \to \R_{+}$ is a smooth and convex penalty function which enforces the measure $\hat{\bmu}^{\gamma}_{n,p}$ to be absolutely continuous. Theoretical properties (such as existence and consistency) of the penalized Wasserstein barycenter $ \hat{\bmu}^{\gamma}_{n,p}$ have been considered in \cite{BCP17}. In this paper, we discuss the choice of the penalty function $E$, as well as the numerical computation of $\hat{\bmu}^{\gamma}_{n,p}$ (using an appropriate discretization of $\Omega$ and a binning of the data), and its benefits for statistical data analysis. 
\begin{rmq}
Note that the restriction of the minimization in \eqref{eq:defmu} to the set $\PP_2(\Omega)$ instead of the whole Wasserstein space  $\PP_2(\R^d)$ is inconsequential.
\end{rmq}

\subsubsection{Fr\'echet mean with respect to a Sinkhorn divergence}

Another way to regularize an empirical Wasserstein barycenter is to use the notion of entropy regularized optimal transportation \cite{cuturi2013sinkhorn,CarlierDPS17} leading to the so-called Sinkhorn divergence
\begin{equation}
\label{eq:Sink_div}
W_{2,\varepsilon}^2(\mu,\nu) = \inf_{\pi\in\Gamma(\mu,\nu)} \iint_{\Omega^2}\vert x-y\vert^2d\pi(x,y)  - \varepsilon \textrm{h} (\pi),
\end{equation}
where $\varepsilon > 0$ is a regularization parameter, and $ \textrm{h}$ stands for the (negative) entropy of the transport plan $\pi$ with respect to the Lebesgue measure on $\Omega \times \Omega$. A regularized Wasserstein barycenter \cite{cuturi2013fast,CuturiPeyre} is then obtained by considering the estimator
\begin{equation}
\hat{\br}^{\varepsilon}_{n,p} = \uargmin{\mu\in\PP_2(\Omega)}  \frac{1}{n}\sum_{i=1}^n W_{2,\varepsilon}^2(\mu,\hat{\nu}_i^{p_i}), \label{eq:defr}
\end{equation}
that can be interpreted as a  Fr\'echet mean with respect to a Sinkhorn divergence and that we call Sinkhorn barycenter. %

\subsection{Contributions}
The selection of the regularisation parameters $\gamma$ or $\epsilon$ is the main issue for computing  adequate penalized or Sinkhorn barycenters in practice. In this paper, we rely  on the Goldenshluger-Lepski (GL) principle in order to perform an automatic calibration of such parameters.

\subsubsection{Data-driven choice of the regularizing parameters}

 The main contribution in this paper is to propose a data-driven choice for the regularization parameters $\gamma$ in \eqref{eq:defmu} and $\varepsilon$ in \eqref{eq:defr}  using the  Goldenshluger-Lepski (GL) method (as formulated in \cite{lacour2016minimal}), which leans on a bias-variance trade-off function, described in details in Section \ref{sec:GL}.
The method consists in comparing estimators pairwise, for a given range of regularization parameters, with respect to a given loss function.
It  provides an optimal regularization parameter that minimizes a bias-variance trade-off function. We displayed in Figure \ref{fig:tradeoff_gaussian_intro} this functional for the dataset of Figure \ref{fig:ex_intro}, which leads to an optimal (in the sense of GL's strategy) parameter choice $\hat{\varepsilon}=1.6$. The entropy regularized Wasserstein barycenter in Figure \ref{fig:ex_intro}(b) is thus chosen accordingly.

\begin{figure}
\centering
\includegraphics[width=0.5 \textwidth,height=0.42\textwidth]{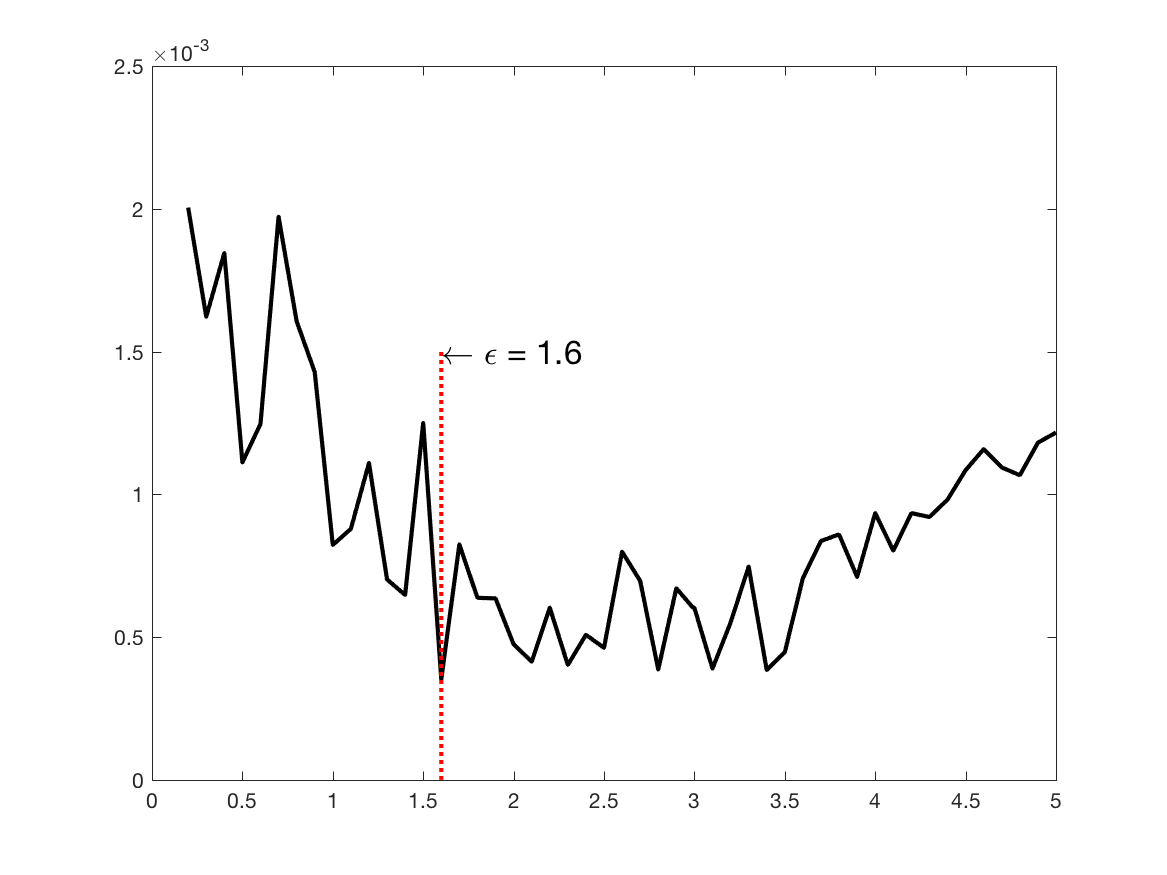}
\caption{The GL's trade-off function associated to the entropy regularized Wassertein barycenters of the dataset in Figure \ref{fig:ex_intro}, for $\varepsilon$ ranging from $0.2$ to $5$}
\label{fig:tradeoff_gaussian_intro}
\end{figure}

 From the results on simulated data displayed in Figure \ref{fig:ex_intro}(b),  it is clear that computing the regularized Wasserstein barycenter $\hat{\br}^{\varepsilon}_{n,p}$ (with an appropriate choice for $\varepsilon$) leads to the estimation of mean density whose shape is consistent with the distribution of the data for each subject. In some sense, the regularization parameters $\gamma$ and $ \varepsilon$ may also be interpreted as the usual bandwidth parameter in kernel density estimation, and their choice greatly influences the shape of the estimators $\hat{\bmu}^{\gamma}_{n,p}$ and $\hat{\br}^{\varepsilon}_{n,p}$ (see Figure  \ref{fig:Sinkhorn_gaussian1D} and Figure \ref{fig:Regbar_gaussian1D}  in Section \ref{sec:expe}).

 To choose the optimal parameter, the GL's strategy  requires 
 an upper bound on the decay to zero of the expected $\L_2(\Omega)$ distance between a regularized empirical barycenter (computed from the data) and its population counterpart. For penalized barycenters \eqref{eq:defmu},   adequate bounds have already been provided in \cite{BCP17}.
 
 \subsubsection{Variance of Sinkhorn estimators}

 To the best of our knowledge, the automatic selection of $\varepsilon$ in the definition of a Sinkhorn divergence has not been considered so far. 
 Another   main contribution of this work then consists in derivating upper bounds on the variance for the estimators 
 $\hat{\br}^{\varepsilon}_{n,p}$ which explicitly depends on 
$\varepsilon$, the number $n$ of  measures, the number $p = \min_{1 \leq i \leq n} p_{i}$ of observations per measures and the size of their support. Such bounds therefore make possible the application of the GL's strategy.

 \subsubsection{Theoretical analysis of the GL's strategy for Sinkhorn barycenters}
 
 For Sinkhorn barycenters, we show that the GL's principle leads to a data-driven choice $\hat{\varepsilon}$ of the regularization parameter which allows to obtain an estimator   $\hat{\br}_{n,p}^{\hat{\varepsilon}}$ that satisfies an oracle inequality implying an optimal trade-off of this estimator between bias and variance terms among a collection of regularization parameters.
 

\subsubsection{Computation issues: binning of the data and discretization of \texorpdfstring{$\Omega$}.}
In our numerical experiments we consider algorithms  for computing  regularized barycenters from a set of discrete measures (or histograms) defined on possibly different grids of points of $\R^{d}$ (or different partitions). 
They are numerical approximations of the regularized Wasserstein barycenters  $\hat{\bmu}_{n,p}^{\gamma}$ and $\hat{\br}_{n,p}^{\varepsilon}$ by a discrete measure of the form $\sum_{k=1}^{N} w_k \delta_{x_k}$  using a fixed grid $\Omega_{N}=\{x_1,\ldots,x_N\}$  of $N$ equally spaced points $x_k\in\mathbb{R}^d$ (bin locations). For simplicity, we adopt a binning of the data \eqref{eq:data} on the same grid, leading to a dataset of discrete measures (with supports included in $\Omega_{N}$) that we denote
\begin{equation}
\label{eq:data_bin}
\tilde{\bq}_i^{p_i}=\frac{1}{p_i}\sum_{j=1}^{p_i}\delta_{\tilde{X}_{i,j}}, \mbox{ where } \tilde{X}_{i,j} = \uargmin{x \in \Omega_{N}} \vert x - X_{i,j} \vert,
\end{equation}
for $1 \leq i \leq n$.
In this paper, we rely on the smooth dual approach proposed in \cite{CuturiPeyre} to compute penalized and Sinkhorn barycenters on a grid of equi-spaced points in $\Omega$ (after a proper binning of the data).

Binning (i.e. choosing the grid $\Omega_{N}$) surely incorporates some sort of additional regularization. A discussion on the influence of the grid size $N$ on the smoothness of the barycenter is proposed in  Section \ref{sec:GL} where we describe the GL's strategy. In our simulations, the choice of $N$ is mainly guided by numerical issues on the computational cost of the algorithms used to approximate  $\hat{\bmu}_{n,p}^{\gamma}$ and $\hat{\br}_{n,p}^{\varepsilon}$.

\subsubsection{Registration of flow cytometry data} \label{sec:flow}
In  biotechnology, flow cytometry is a high-throughput technique that can measure a large number  of surface and intracellular markers of single cell in a biological sample. With this technique, one can assess individual characteristics (in the form of multivariate data) at a cellular level to determine the type of cell, their functionality and the way they differ. At the beginning of flow cytometry, the analysis of such data was performed manually by visually separating regions or gates of interest on a series of sequential bivariate projection of the data, a process known as gating. However, the development of this technology now leads to datasets made of multiple measurements (e.g.\ up to 18) of millions of individuals cells. A significant amount of work has thus been carried out in recent years to propose automatic statistical methods to overcome the limitations of manual gating  (see e.g.\ \cite{hahne2010per,hejblum2017sequential,lee2016modeling,pyne2014joint} and references therein).

When analyzing samples in cytometry measured from different patients, a critical issue is data registration across patients. As carefully explained in \cite{hahne2010per}, the alignment of flow cytometry data is a preprocessing step which aims at removing effects coming from technological issues in the acquisition of the data rather than significant biological differences. In this paper, we use data analyzed in \cite{hahne2010per} that are obtained from a renal transplant retrospective study conducted by the Immune Tolerance Network (ITN). 
This dataset is freely available from the {\tt flowStats} package of Bioconductor \cite{Gentleman2004} that can be downloaded from
\url{http://bioconductor.org/packages/release/bioc/html/flowStats.html}. It consists of samples from $15$ patients.

 After an appropriate scaling trough  an arcsinh transformation and  an initial gating on total lymphocytes to remove artefacts, we focus our analysis on the cell markers FSC (forward-scattered light) and SSC (side-scattered light) which are of interest to measure the volume and morphological complexity of cells. The number of considered cells by patient varies from $88$ to $2185$. The resulting dataset is displayed in Figure \ref{fig:ex_cytometry2D}. It clearly shows a mis-alignment issue between measurements from different patients.

  The last contribution of the paper is thus to demonstrate the usefulness of regularized Wasserstein barycenters to correct mis-alignment effects in the  analysis of data produced by flow cytometers. 

\begin{figure}[h]
\centering
\includegraphics[width=0.85 \textwidth,height=0.65\textwidth]{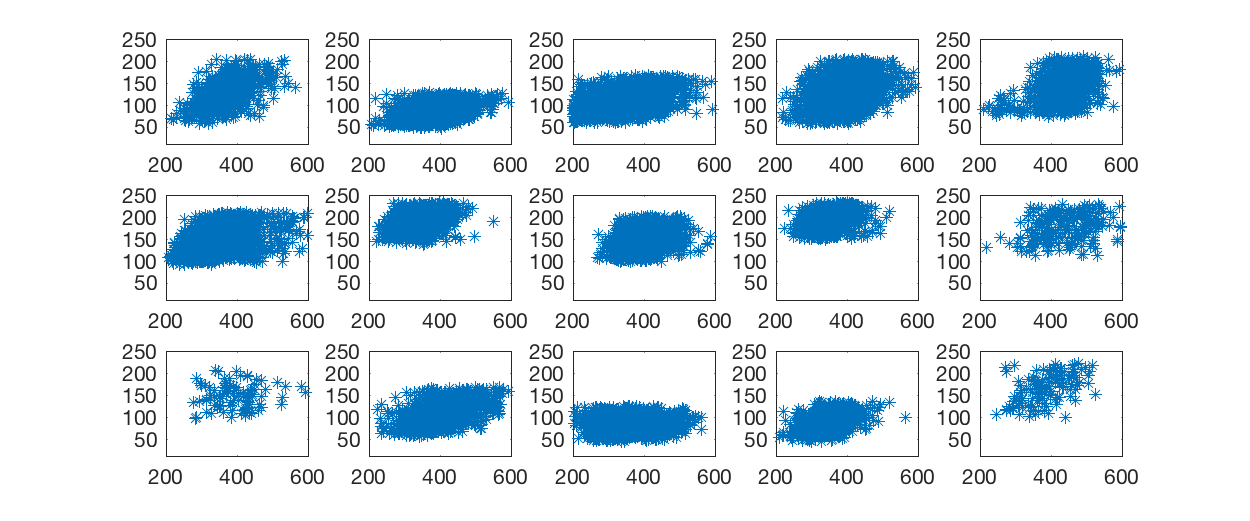} 
\caption{Example of  flow cytometry data measured from $n=15$ patients (restricted to a bivariate projection).  The horizontal axis (resp.\  vertical axis) represent the values of the FSC (resp.\ SSC) cell marker.}
\label{fig:ex_cytometry2D}
\end{figure}

\subsection{Organization of the paper}

The analysis of the variance of the regularized Wasserstein barycenters $\hat{\bmu}^{\gamma}_{n,p}$ and $\hat{\br}^{\varepsilon}_{n,p}$ are detailed in Section \ref{sec:Regbar} and Section \ref{sec:Sinkhorn} respectively. Section \ref{sec:GL} contains a description of the Goldenshluger-Lepski principle to choose the regularization parameters $\gamma$ and $\varepsilon$ as well as an oracle inequality justifying this technique for Sinkhorn barycenters. Section \ref{sec:expe} finally reports the results from  numerical experiments using simulated data and the flow cytometry dataset displayed in Figure \ref{fig:ex_cytometry2D}. 
We conclude the paper in Section \ref{sec:conclusion} by a brief discussion.
Some proofs and technical results are presented in Appendix \ref{sec:strong_convexity_H},  Appendix \ref{sec:dual_bounded}, and Appendix \ref{app:concen}. Algorithmic details on the computation of the estimator $\hat{\bmu}^{\gamma}_{n,p}$ are gathered in  Appendix \ref{sec:algo}.

\section{Penalized Wasserstein barycenters}\label{sec:Regbar}
In this section, we adopt the framework from \cite{BCP17} in which the Wasserstein barycenter is regularized through a convex penalty function as presented in \eqref{eq:defmu}.

\subsection{Minimization problem and variance properties}
Let us remind some of the basic definitions in  \cite{BCP17}. We define $W_2(\PP_2(\Omega))$ as the space of distributions $\P$ supported on $\PP_2(\Omega)$. The penalized Wasserstein barycenter associated to the distribution $\P$ is defined as a solution of the minimization problem
$$ \min_{\mu\in\PP_2(\Omega)}\ \int_{\PP_2(\Omega)}W_2^2(\mu,\nu)d\P(\nu) +\gamma E(\mu)$$
where $\gamma > 0$ is a penalization parameter and the penalty function writes
\begin{equation}
\label{ex_penalty}
E(\mu) = \left\{\begin{array}{ll}
\Vert f\Vert_{H^{k}(\Omega)}^{2}, & \mbox{if}\ f =\frac{d\mu}{dx} \ \mbox{and} \ f\geq \alpha, \\
+\infty & \mbox{otherwise.}
\end{array}\right.
\end{equation}
where $\Vert\cdot\Vert_{H^k(\Omega)}$ denotes the Sobolev norm associated to the $\mathbb{L}^2(\Omega)$ space,  $\alpha>0$ is arbitrarily small and $k>d-1$.
Remark that the  function $E$ is strictly convex on its domain
\begin{equation}\label{eq:domainE}
\mathcal{D}(E)=\{\mu\in\PP_2(\Omega)\ \mbox{such that} \ E(\mu)<+\infty\}.
\end{equation}
This choice is supported by the discussion in Section 5 of \cite{BCP17}, that in particular imposes the barycenter $\mu$ to belong to $\PP_2^{ac}(\Omega)$, the space of measures in $\PP_2(\Omega)$ that are absolutely continuous. It is mainly driven by the need to retrieve an absolutely continuous measure from discrete observations $(X_{ij})$, as it is often done when approximating data through kernel smoothing in density estimation. Others examples of penalty functions are given in \cite{BCP17}, including a class of relative G-functional described in Section 9.4 in \cite{ambrosio2008gradient}.

\begin{defi}\label{def:pen_bar}
Let $\bnu_1,\ldots,\bnu_n \in\PP_2(\Omega)$ be iid measures with distribution $\P$ such that $\P(\PP_2^{ac}(\Omega))>0$. The empirical probability measure defined by $(\bnu_i)_{i=1,\ldots,n}$ writes $\P_n=\frac{1}{n}\sum_{i=1}^n\delta_{\bnu_i}$. Let us then consider the random measures $(\hat{\bnu}_i^{p_i})_{1\leq i\leq n}$ of the form $ \hat{\bnu}_i^{p_i}=\frac{1}{p_i}\sum_{j=1}^{p_i}\delta_{X_{i,j}}$ where $(X_{i,j})_{1\leq j\leq p_i}$ are iid random variables of law $\bnu_i$. From there on we define for $\gamma>0$ the barycenters:
\begin{align}
\mu^{\gamma} &= \uargmin{\mu\in\PP_2(\Omega)} \ \int_{\PP_2(\Omega)}W_2^2(\mu,\nu)d\P(\nu)+\gamma E(\mu)= \E_{\bnu\sim\P}[W_2^2(\mu,\bnu) ]+\gamma E(\mu)\label{def:population_Regbar}\\
\hat{\bmu}^{\gamma}_{n,p} &= \uargmin{\mu\in\PP_2(\Omega)} \ \int_{\PP_2(\Omega)}W_2^2(\mu,\nu)d\P_n(\nu)+\gamma E(\mu)= \frac{1}{n}\sum_{i=1}^nW_2^2(\mu,\hat{\bnu}_i^{p_i})+\gamma E(\mu) \label{def:simulated_Regbar}
\end{align}
called respectively penalized population barycenter \eqref{def:population_Regbar} and penalized empirical barycenter \eqref{def:simulated_Regbar}.
\end{defi}
In \cite{BCP17}, the  existence and uniqueness of these barycenters have been shown for $\gamma>0$. This holds true for measures defined on $\PP_2(\R^d)$. By penalizing the barycenters with function $E$ as in \eqref{ex_penalty}, we enforce them to be absolutely continuous.  Therefore,  let  $\hat{\bfun}^{\gamma}_{n,p}$ and $f^{\gamma}_{\P}$ be  the densities  associated to $\hat{\bmu}^{\gamma}_{n,p}$ and $\mu^{\gamma}_{\P}$.

In order to apply the GL's strategy, 
we now study the expected squared  $\L_2(\Omega)$-distance $\E(\Vert \hat{\bfun}^{\gamma}_{n,p}-f^{\gamma}_{\P}\Vert^2)$ that will be referred to as a ``variance term'' for $\hat{\bmu}^{\gamma}_{n,p}$. 

\begin{thm}[Section 5 in \cite{BCP17}]\label{th:rate_Regbar}
Let $\Omega\subset\R^d$ be compact and $\hat{\bfun}_{n,p}^{\gamma}$ and $f_{\P}^{\gamma}$ be the density functions of $\hat{\bmu}_{n,p}^{\gamma}$ and $\mu_{\P}^{\gamma}$, induced by the choice \eqref{ex_penalty} of the penalty function $E$. Then, provided that $d<4$, there exists a constant $c>0$ depending only on $\Omega$ such that
\begin{equation}\label{eq:rate_Regbar}
\E\left(\Vert \hat{\bfun}_{n,p}^{\gamma}-f_{\P}^{\gamma}\Vert^2_{\mathbb{L}_2(\Omega)}\right)\leq  c \left(\frac{1}{\gamma p^{1/4}}+\frac{1}{\gamma n^{1/2}}\right)
\end{equation}
where $p = \min_{1 \leq i \leq n} p_{i}$.
\end{thm}

Thanks to this result, we will be able to automatically calibrate the parameter $\gamma>0$ by following  the GL's parameter selection strategy described in Section \ref{sec:GL}. 
\begin{rmq}
 As already mentioned above, in practice, we discretize $\Omega$ into a sufficiently fine and fixed grid on which we compute a discretized version the penalized Wasserstein barycenter $\hat{\bfun}_{n,p}^{\gamma}$. Therefore, the upper bound \eqref{eq:rate_Regbar} should be considered with some caution as it is not exactly a control of the variance of the discretized estimator that is used in our numerical experiments. Moreover, the upper bound \eqref{eq:rate_Regbar} involves a constant $c >0$ whose derivation is guided by theoretical arguments. However, a good calibration of $c$ is of primary importance as discussed in the numerical experiments reported in Section \ref{sec:expe}. 
 \end{rmq}

\subsection{Numerical computation} As a practical complement to \cite{BCP17}, we provide in  Appendix \ref{sec:algo}  efficient minimization algorithms   for the computation of $\hat{\bfun}_{n,p}^{\gamma}$, after a binning of the data on a fixed grid $\Omega_N$. For $\Omega$ included in the  real line, a simple subgradient descent is considered. When data are histograms supported on $\R^d$, $d\geq 2$, we rely on a smooth dual approach  based on the work of \cite{CuturiPeyre}.
\section{Sinkhorn barycenters via entropy regularized optimal transport}\label{sec:Sinkhorn}

In this section, we analyze the variance of the Sinkhorn barycenter  defined in \eqref{eq:defr}.

\subsection{Variance properties of the Sinkhorn barycenters}

As before we consider a binning of the  data on a fixed and finite discrete grid $\Omega_N$. 
For two discrete measures $r,q\in\Sigma_N$, the Sinkhorn divergence \eqref{eq:Sink_div} reads for  $\varepsilon>0$
\begin{equation}\label{def:primal_Sinkhorn}
W_{2,\varepsilon}^2(r,q):=\underset{U\in U(r,q)}{\min}\  \langle C, U\rangle -\varepsilon \textrm{h}(U).
\end{equation}
where the discrete (negative) entropy for a given coupling $U\in U(r,q)$ is given by $ \textrm{h}(U):=-\sum_{m, \ell}U_{m \ell}\log U_{m \ell}$.  We shall then use two key properties to analyze the variance of Sinkhorn barycenters which are  the strong convexity (see Theorem \ref{th:strong_convexity_H} below) and the Lipschitz continuity (see Lemma \ref{lemma:Lipschitz} below) of the mapping $r \mapsto W_{2,\varepsilon}^2(r,q)$ (for a given $q \in \Sigma_N$).

However, to guarantee the  Lipschitz continuity of this mapping, it is necessary to restrict the analysis to discrete measures $r$  belonging to the convex set
$$
\Sigma_N^{\rho} = \left\{ r \in \Sigma_N \; : \; \min_{1 \leq \ell \leq N} r_{\ell} \geq \rho \right\},
$$
where $0 < \rho < 1$ is an arbitrarily small constant. This means that our theoretical results on the variance  of the Sinkhorn barycenters  hold for discrete measures with non-vanishing entries. Nevertheless, we obtain upper bounds on these variances which depend explicitly on the constant $\rho$, allowing to discuss its choice.

Then, as it has been done for the penalized barycenters in Definition \ref{def:pen_bar}, we introduce the definitions of empirical and population Sinkhorn barycenters. 

\begin{defi}\label{def:Sinkhorn_bar}
Let $0 < \rho < 1/N$, and $\P$ be a probability distribution on $\Sigma_N^{\rho}$. Let $\bq_1,\ldots,\bq_n\in\Sigma_N^{\rho}$ be an iid sample drawn from the distribution $\P$. For each $1 \leq i \leq n$, we assume that $(\tilde{X}_{i,j})_{1\leq j\leq p_i}$ are iid random variables sampled from $\bq_i$.  For each $1 \leq i \leq n$, let us define the following discrete measures
$$
\tilde{\bq}_i^{p_i} = \frac{1}{p_i}\sum_{j=1}^{p_i}\delta_{\tilde{X}_{i,j}} \quad \mbox{and} \quad \hat{\bq}_i^{p_i}  = (1 - \rho N)\tilde{\bq}_i^{p_i} +  \rho \1_{N},
$$
where $\1_{N}$ is the vector of $\R^N$ with all entries equal to one. Thanks to the condition $0 < \rho < 1/N$, it follows  that $\hat{\bq}_i^{p_i} \in \Sigma_N^{\rho}$ for  all $1 \leq i \leq n$, which may not be the case for some $\tilde{\bq}_i^{p_i},\ 1=1,\ldots,n$. Then, we define
\begin{align}
r^{\varepsilon} &=\uargmin{r\in\Sigma_N^{\rho}}\E_{\bq\sim\P}[W_{2,\varepsilon}^2(r,\bq)] & \qquad \mbox{the population Sinkhorn  barycenter}\label{def:pop_emp_Sinkhorn}\\
\hat{\br}^{\varepsilon}_{n,p} &=\uargmin{r\in\Sigma_N^{\rho}}\frac{1}{n}\sum_{i=1}^nW_{2,\varepsilon}^2(r,\hat{\bq}_i^{p_i}) & \qquad \mbox{the empirical Sinkhorn  barycenter}\label{def:simulated_Sinkhorn}
\end{align}
\end{defi}
In the optimisation problem \eqref{def:simulated_Sinkhorn}, we choose to use the discrete measures $\hat{\bq}_i^{p_i}$  instead of the empirical measures $\tilde{\bq}_i^{p_i}$ to guarantee the use of discrete measures belonging to $\Sigma_N^{\rho}$ in the definition of  the empirical Sinkhorn barycenter $\hat{\br}^{\varepsilon}_{n,p}$.

\begin{rmq}
The population and empirical Sinkhorn  barycenters in Definition \ref{def:Sinkhorn_bar} are constrained to belong to the set $\Sigma_N^{\rho}$
 so that the Lipschitz continuity of $r \mapsto W_{2,\varepsilon}^2(r,q)$  holds true.
\end{rmq}

The following theorem is the main result of this section which gives an upper bound on 
$\E(\vert r^{\varepsilon}-\hat{r}^{\varepsilon}_{n,p}\vert^2)$  which will be referred to as the variance of $\hat{\br}^{\varepsilon}_{n,p}$ in what follows.
 In particular, the asymptotic regime in which we are interested in is the number $n$ of measures defining the barycenter as well as the number $p$ of observations per measures.
\begin{thm}\label{th:Sinkhorn_rate}
Recall that $p=\min_{1 \leq i \leq n} p_i$ and let $\varepsilon >0$. Then, one has that
\begin{equation}\label{eq:Sinkhorn_rate}
\E(\vert r^{\varepsilon}-\hat{r}^{\varepsilon}_{n,p}\vert^2) \leq  \frac{16 L_{\rho,\varepsilon}^2}{\varepsilon^2n}  +  \frac{2 L_{\rho,\varepsilon}}{\varepsilon} \left(   \sqrt{\frac{N}{ p}} + 2 \rho( N + \sqrt{N} ) \right),
\end{equation}
with
\begin{equation} 
L_{\rho,\varepsilon}= \left( \sum_{1 \leq m \leq N} \left(2 \varepsilon \log(N) +\adjustlimits\inf_{1 \leq k \leq N} \sup_{1 \leq \ell\leq N}|C_{m\ell} - C_{k \ell}|    -  \varepsilon \log(\rho) \right)^2  \right)^{1/2}. \label{eq:defL}
\end{equation}
\end{thm}

A few remarks can be made about the above result.  
The bound in the right-hand side of \eqref{eq:Sinkhorn_rate}  explicitly depends  on the size $N$ of the grid. This  will  be taken into account for the choice of the optimal parameter $\hat{\varepsilon}$ (see Section \ref{sec:GL}). Moreover, it can be used to discuss the choice of $\rho$. First, if one takes  $\rho = \epsilon^{\kappa}$, the Lipschitz constant (Lemma \ref{lemma:Lipschitz}) 
$L_{\rho,\varepsilon}=L_{\varepsilon}$ becomes
$$
L_{\varepsilon} = \left( \sum_{1 \leq m \leq N} \left( \varepsilon (2\log(N)-\kappa \log(\epsilon)) + \adjustlimits\inf_{1 \leq k \leq N} \sup_{1 \leq \ell\leq N} |C_{m\ell} - C_{k \ell}| \right)^2  \right)^{1/2}, \label{eq:defLeps}
$$
which is a constant (not depending on $\rho$) such that
$$
\lim_{\epsilon \to 0} L_{\varepsilon}  = \left( \sum_{1 \leq m \leq N} \left(  \adjustlimits\inf_{1 \leq k \leq N} \sup_{1 \leq \ell\leq N} |C_{m\ell} - C_{k \ell}| \right)^2  \right)^{1/2}.
$$
If we further assume that $\rho = \epsilon^{\kappa} < \min(1/N,1/p)$ we obtain the upper bound
\begin{equation}\label{eq:Sinkhorn_rate2}
\E(\vert r^{\varepsilon}-\hat{r}^{\varepsilon}_{n,p}\vert^2) \leq  \frac{16 L_{\varepsilon}^2}{\varepsilon^2n}  +  \frac{2 L_{\varepsilon}}{\varepsilon} \left(   \sqrt{\frac{N}{ p}} + 2 \left( \frac{N}{p} + \sqrt{\frac{N}{p^2}} \right) \right).
\end{equation}
Therefore, choosing $\varepsilon = \varepsilon(n,p)\underset{\min(n,p) \to\infty}{\longrightarrow} 0$ such that $1/\varepsilon^2n\to\infty$ and $1/\varepsilon p\to \infty$, we have that $\E(\vert r^{\varepsilon}-\hat{r}^{\varepsilon}_{n,p}\vert^2)\underset{n,p\to\infty}{\longrightarrow}$ (take for example $\varepsilon=1/n^{\alpha}p^{\beta}$ with $0<\alpha<1/2$ and $0<\beta<1$).

Finally, it should be remarked that Theorem \ref{th:Sinkhorn_rate} holds for general cost matrices $C$ that are symmetric and non-negative.

\subsection{Proof of Theorem \ref{th:Sinkhorn_rate}}\label{sec:StrongConvexity}

The proof of the upper bound \eqref{eq:Sinkhorn_rate} relies on   strong convexity of the functional $r\mapsto W_{2,\varepsilon}^2(r,q)$ for $q\in\Sigma_N$, without constraint on its entries. 
This property can be derived by studying the Legendre transform of $r\mapsto W_{2,\varepsilon}^2(r,q)$. For a fixed distribution $q\in\Sigma_N$, using the notation in \cite{CuturiPeyre}, we define the function
$$H_q(r) :=W_{2,\varepsilon}^2(r,q), \quad \mbox{for all} \ r\in\Sigma_N.$$
Its Legendre transform is given for $g\in\R^N$ by $H_q^{\ast}(g)=\underset{r\in\Sigma_N}{\max} \ \langle g,r \rangle - H_q(r)$ and its differentiation properties are presented in the following theorem.
\begin{thm}[Theorem 2.4 in \cite{CuturiPeyre}]
\label{th:diffHq*}
For $\varepsilon>0$, the Legendre dual function $H_q^{\ast}$ is $C^{\infty}$. Its gradient function $\nabla H_q^{\ast}$ is $1/\varepsilon$-Lipschitz. Its value, gradient and Hessian at $g\in\R^N$ are, writing $\alpha=\exp(g/\varepsilon)$ and $K=\exp(-C/\varepsilon)$,
\begin{align*}
H_q^{\ast}(g)&=\varepsilon(E(q)+\langle q,\log (K\alpha)\rangle), \ \nabla H_q^{\ast}(g)=\diag(\alpha)K\frac{q}{K\alpha}\in\Sigma_N\\
\nabla^2H_q^{\ast}(g) &=\frac{1}{\varepsilon}\left(\diag\left(\diag(\alpha)K\frac{q}{K\alpha}\right)\right)-\frac{1}{\varepsilon}\diag(\alpha)K\diag\left(\frac{q}{(K\alpha)^2}\right)K\diag(\alpha),
\end{align*}
where the notation $\frac{q}{r}$  stands for the component-wise division of the entries of $q$ and $r$.
\end{thm}
From this result, we can deduce the strong convexity of the dual functional $H_q$ as stated below.
\begin{thm}\label{th:strong_convexity_H}
Let $\varepsilon>0$. Then, for any $q\in\Sigma_N$, the function $H_q$ is $\varepsilon$-strongly convex for the Euclidean 2-norm.
\end{thm}

The proof of Theorem \ref{th:strong_convexity_H} is deferred to Appendix \ref{sec:strong_convexity_H}. 
We can also ensure the Lipschitz continuity of $H_q(r)$, when restricting our analysis to the set $r\in\Sigma_N^{\rho}$.
\begin{lemma}\label{lemma:Lipschitz}
Let $q \in \Sigma_N$ and $0 < \rho < 1$. Then, one has that $r \mapsto H_q(r)$ is $L_{\rho,\varepsilon}$-Lipschitz on $\Sigma_N^{\rho}$ with $L_{\rho,\varepsilon}$ defined in \eqref{eq:defL}.
\end{lemma}
The proof of this Lemma is given in Appendix \ref{sec:dual_bounded}.

We can now proceed to the proof of Theorem \ref{th:Sinkhorn_rate}. Let us introduce the following Sinkhorn barycenter
\begin{equation}
\br_n^{\varepsilon} =\uargmin{r\in\Sigma_N^{\rho}}\frac{1}{n}\sum_{i=1}^n W_{2,\varepsilon}^2(r,{\bq}_i) = \uargmin{r\in\Sigma_N^{\rho}}\frac{1}{n}\sum_{i=1}^n H_{{\bq}_i}(r) , \label{eq:defrneps}
\end{equation}
of the iid random measures $\bq_1,\ldots,\bq_n$
sampled from the distribution $\P$ supported on  $\Sigma_N^{\rho}$.
By the triangle inequality, we have that
\begin{equation}\label{ineq:rate_Sinkhorn}
\E(\vert r^{\varepsilon}-\hat{\br}^{\varepsilon}_{n,p}\vert^2)\leq 2\E(\vert {r}^{\varepsilon}-{\br}_n^{\varepsilon}\vert^2)+2\E(\vert {\br}_n^{\varepsilon}-\hat{{\br}}^{\varepsilon}_{n,p}\vert^2).
\end{equation}
To control the first term of the right hand side of the above inequality, we use that (for any $q \in \Sigma_N$) $r \mapsto H_q(r)$ is $\varepsilon$-strongly convex by Theorem \ref{th:strong_convexity_H} and $L_{\rho,\varepsilon}$-Lipschitz on $\Sigma_N^{\rho}$ by Lemma \ref{lemma:Lipschitz} where $L_{\rho,\varepsilon}$ is the constant defined by equation \eqref{eq:defL}. Under these assumptions, it follows from the arguments in the proof of Theorem 6 in \cite{shalev2009stochastic} that
\begin{equation}\label{conv:Sinkhorn}
\E(\vert r^{\varepsilon}-{\br}_n^{\varepsilon}\vert^2)\leq \frac{8 L^2_{\rho,\varepsilon}}{\varepsilon^2n}.
\end{equation}
The strong convexity of $H_q$ has a major role here as it brings out the distance between the empirical minimizer ${\br}_n^{\varepsilon}$ and any other point in $\Sigma_N$.
For the second term in the right hand side of \eqref{ineq:rate_Sinkhorn}, we obtain by the strong convexity of $H_q$ that
$$\frac{1}{n}\sum_{i=1}^n H_{{\bq}_i}(\hat{{\br}}^{\varepsilon}_{n,p})\geq \frac{1}{n}\sum_{i=1}^n H_{{\bq}_i}({\br}_n^{\varepsilon})+\frac{1}{n}\sum_{i=1}^n \nabla H_{{\bq}_i}({\br}_n^{\varepsilon})^T(\hat{\br}^{\varepsilon}_{n,p}-{\br}_n^{\varepsilon})+\frac{\varepsilon}{2}\vert {\br}_n^{\varepsilon}-\hat{{\br}}^{\varepsilon}_{n,p}\vert^2.$$
Theorem 3.1 in \cite{CuturiPeyre} ensures that $\frac{1}{n}\sum_i\nabla H_{q_i}({\br}_n^{\varepsilon})=0$. The same inequality also holds for the terms $H_{\hat{{\bq}}_i^{p_i}}$, and we therefore have
\begin{align*}
& \frac{1}{n}\sum_{i=1}^n H_{{\bq}_i}(\hat{{\br}}^{\varepsilon}_{n,p})\geq \frac{1}{n} \sum_{i=1}^n H_{{\bq}_i}({\br}_n^{\varepsilon})+\frac{\varepsilon}{2}\vert {\br}_n^{\varepsilon}-\hat{{\br}}^{\varepsilon}_{n,p}\vert^2,\\
& \frac{1}{n}\sum_{i=1}^n H_{\hat{{\bq}}_i^{p_i}}({\br}_n^{\varepsilon})\geq \frac{1}{n}\sum_{i=1}^n H_{\hat{{\bq}}_i^{p_i}}(\hat{{\br}}^{\varepsilon}_{n,p})+\frac{\varepsilon}{2}\vert {\br}_n^{\varepsilon}-\hat{{\br}}^{\varepsilon}_{n,p}\vert^2.
\end{align*}
Using the symmetry of the Sinkhorn divergence, Lemma \ref{lemma:Lipschitz} also implies that the mapping  $q \mapsto H_q(r)$ is $L_{\rho,\varepsilon}$-Lipschitz on $\Sigma_N^{\rho}$ for any discrete distribution $r$. Hence,  by summing the two above inequalities, and by taking the expectation on both sides, we obtain that
$$
\varepsilon\E(\vert {\br}_n^{\varepsilon}-\hat{{\br}}^{\varepsilon}_{n,p}\vert^2)  \leq \frac{2L_{\rho,\varepsilon}}{n}\sum_{i=1}^n\E(\vert \bq_i-\hat{\bq}_i^{p_i}\vert).
$$
Using the inequalities 
$$
\vert \bq_i-\hat{\bq}_i^{p_i}\vert \leq \vert \bq_i-\tilde{\bq}_i^{p_i}\vert  + \rho N \vert \tilde{\bq}_i^{p_i} \vert  + \rho  \vert \1_{N}\vert \leq  \vert \bq_i-\tilde{\bq}_i^{p_i}\vert  +  \rho( N + \sqrt{N}),
$$
we finally have that
\begin{align}\label{relation1}
\varepsilon\E(\vert {\br}_n^{\varepsilon}-\hat{{\br}}^{\varepsilon}_{n,p}\vert^2) & \leq 2L_{\rho,\varepsilon} \left(  \frac{1}{n}\sum_{i=1}^n\sqrt{\E(\vert \bq_i-\tilde{\bq}_i^{p_i}\vert^2)} + \rho( N + \sqrt{N} ) \right).
\end{align}
Conditionally on $\bq_i$, one has that $p_i \tilde{\bq}_i^{p_i}$ is a  random vector following a multinomial distribution $\mathcal{M}(p_i,\bq_i)$. Hence, for each $1 \leq k \leq N$, denoting $\bq_{i,k}$ (resp. $\tilde{\bq}_{i,k}^{p_i}$) the $k$-th coordinate of $\bq_{i}$ (resp. $\tilde{\bq}_{i}^{p_i}$), one has that 
$$
\E \left(  \tilde{\bq}_{i,k}^{p_i}  | \bq_i \right)= \bq_{i,k} \quad \mbox{and} \quad \E \left[ \left(  \tilde{\bq}_{i,k}^{p_i} - \bq_{i,k} \right)^2  | \bq_i  \right]=\frac{ \bq_{i,k} (1- \bq_{i,k})}{ p_i } \leq \frac{ 1}{ 4 p_i }.
$$
Thus, we have 
\begin{align}\label{relation2}
\E(\vert \bq_i-\tilde{\bq}_i^{p_i}\vert^2) = \sum_{k=1}^{N}  \E ( \bq_{i,k}-\tilde{\bq}_{i,k}^{p_i})^2 \leq \frac{1}{4} \sum_{k=1}^{N}  p_{i}^{-1} \leq \frac{N}{4 p}
\end{align}
and we obtain  from \eqref{relation1} and \eqref{relation2} that
\begin{equation} \label{conv:Sinkhorn2}
\E(\vert {\br}_n^{\varepsilon}-\hat{{\br}}^{\varepsilon}_{n,p}\vert^2)     \leq \frac{L_{\rho,\varepsilon}}{\varepsilon} \left(   \sqrt{\frac{N}{ p}} + 2 \rho( N + \sqrt{N} ) \right).
\end{equation}
Combining inequalities \eqref{ineq:rate_Sinkhorn}, \eqref{conv:Sinkhorn}, and \eqref{conv:Sinkhorn2} concludes the proof of Theorem \ref{sec:StrongConvexity}.

\section{Goldenshluger-Lepski method and oracle inequality} \label{sec:GL}

In this section, we present a method to choose, in a data-driven way, the parameters $\gamma$ in \eqref{eq:defmu} and $\varepsilon$ in \eqref{eq:defr}. By analogy with the work in \cite{lacour2016minimal} based on the Goldenshluger-Lepski (GL) principle \cite{goldenshluger2008universal}, we propose to compute a bias-variance trade-off functional which will provide an automatic selection method for the regularization parameters within a  finite set for either penalized or Sinkhorn barycenters. The method consists in comparing estimators pairwise, for a given range of regularization parameters, with respect to a loss function.

Since the formulation of the GL's principle is similar for both estimators,  we  only present its principle for the Sinkhorn barycenter described in Section \ref{sec:Sinkhorn}.  We also show that the formulation of the GL's principle proposed in this paper for the Sinkhorn barycenter leads to a data-driven estimator satisfying an oracle inequality which sheds some light on its theoretical properties.

The key point in the GL method is the definition of a data-driven trade-off functional that is composed of a term measuring the disparity  between two estimators and of a penalty term that is chosen according to the upper bounds on the variance of the Sinkhorn barycenter given in Section \ref{sec:Sinkhorn}. More precisely, we assume that we have at our disposal a collection of estimators $(\hat{\br}_{n,p}^{\varepsilon})_{\varepsilon}$ for $\varepsilon$ ranging in a finite set $\Lambda \subset \R_+$ depending on the data at hand. The GL method consists in choosing a value $\hat{\varepsilon}$ which minimizes the following bias-variance trade-off function:
\begin{equation}
\label{eq:GL_method}
\hat{\varepsilon}= \uargmin{\varepsilon\in\Lambda}\ B(\varepsilon)+3V(\varepsilon)
\end{equation}
for which we set the ``bias term'' as
\begin{equation}
\label{eq:B_tradeoff}
B(\varepsilon)=\sup_{\tilde{\varepsilon}\leq\varepsilon} \ \left[\vert \hat{\br}_{n,p}^{\varepsilon}-\hat{\br}_{n,p}^{\tilde{\varepsilon}}\vert^2 - 3 V(\tilde{\varepsilon})\right]_{+}
\end{equation}
where $x_{+}=\max(x,0)$ denotes the positive part,   
and the ``variance term'' $V$ is chosen accordingly to the oracle inequality in the Theorem \eqref{theo:oracle} below as follows
\begin{equation}
\label{eq:var_bound}
V(\varepsilon) = V_{b_1,b_2}(\varepsilon) :=  b_1 \frac{8 L_{\rho,\varepsilon}^2}{\varepsilon^2n}  +  \frac{2 L_{\rho,\varepsilon}}{\varepsilon} \left(   \sqrt{b_2 \frac{N}{ p}} +  \rho( N + \sqrt{N} ) \right),
\end{equation}
where $b_1 > 0$ and $b_2 >0$ are constants whose choice is discussed below. Note that, in definition  \eqref{eq:B_tradeoff} of the bias term, it is implicitly understood that the supremum is restricted to the regularization parameters $\tilde{\varepsilon}\leq\varepsilon$ such that $\tilde{\varepsilon} \in \Lambda$. To stress the dependence of the variance term on $b_1$ and $b_2$, we sometimes write $V(\varepsilon) = V_{b_1,b_2}(\varepsilon)$.

Following \cite{lacour2016minimal}, we propose to show that, under an appropriate choice of the constants $b_1$ and $b_2$, the selected estimator $\hat{\br}_{n,p}^{\hat{\varepsilon}}$ satisfies  an oracle inequality which represents an optimal bias-variance tradeoff (depending on the set $\Lambda$) for its risk $\vert \hat{\br}_{n,p}^{\hat{\varepsilon}}-r^{0}\vert^2$, where
\begin{equation}\label{def:r0}
r^{0} \in \uargmin{r\in\Sigma_N^{\rho}}\E_{\bq\sim\P}[W_{2,0}^2(r,\bq)] \quad \mbox{with} \quad
W_{2,0}^2(r,q):=\underset{U\in U(r,q)}{\min}\  \langle C, U\rangle.
\end{equation}

\begin{rmq}
It should be remarked that we chose to refer to $\E(\vert r^{\varepsilon}-\hat{r}^{\varepsilon}_{n,p}\vert^2)$  as a ``variance term''. This is somewhat imprecise as the estimator $\hat{r}^{\varepsilon}_{n,p}$ is certainly such that $\E(\hat{r}^{\varepsilon}_{n,p}) \neq r^{\varepsilon}$. Therefore, $\E(\vert r^{\varepsilon}-\hat{r}^{\varepsilon}_{n,p}\vert^2)$ is not the usual statistical notion of variance for $\hat{r}^{\varepsilon}_{n,p}$. Similarly, we have chosen to implicitly referred to $\vert r^{\varepsilon} - r^{0}\vert$ as a bias term which is rather an approximation error term. Nevertheless, we prefer to keep this terminology of bias and variance as it is consistent with the one used to present the GL's principle in \cite{lacour2016minimal} for the classical problem of kernel density estimation for which the standard notions of a bias term and an approximation error coincide.
\end{rmq}

Now, as in \cite{lacour2016minimal}, we introduce the following generalized approximation error 
$$
D(\varepsilon) := \max \left(  \sup_{\tilde{\varepsilon}\leq\varepsilon} \ \vert r^{\tilde{\varepsilon}}-r^{\varepsilon}\vert  , \vert r^{\varepsilon}-r^{0}\vert \right)  
$$
which satisfies $D(\varepsilon)  \leq 2  \sup_{\tilde{\varepsilon}\leq\varepsilon} \vert r^{\tilde{\varepsilon}}-r^{0}\vert$. The following result shows that, under an appropriate choice of $b_1$ and $b_2$, the data-driven choice $\hat{\varepsilon}$ of the regularization parameter by the GL method leads to a Sinkhorn barycenter $\hat{\br}_{n,p}^{\hat{\varepsilon}}$ satisfying an oracle inequality leading to an optimal bias-variance tradeoff within the collection $\Lambda$ of regularization parameters. 

\begin{thm} \label{theo:oracle} 
Assume that the constants $b_1$ and $b_2$ in the calibration of the variance term  $V(\varepsilon) = V_{b_1,b_2}(\varepsilon)$ are such that
\begin{equation}
\label{def:b1b2}
b_1 > (1 + \sqrt{\log(|\Lambda|/2)})^2 \quad \mbox{and} \quad  b_2 >  \log(2)+ \frac{\log(n)+\sqrt{\log(|\Lambda|/2)}}{N}.
\end{equation}
where $|\Lambda|$ denotes the cardinal of $\Lambda$. Then, one has that
\begin{equation}
\vert \hat{\br}_{n,p}^{\hat{\varepsilon}}-r^{0}\vert \leq (1 + 2 \sqrt{3} ) \inf_{\varepsilon \in \Lambda} \left\{ D(\varepsilon) +  \sqrt{V_{b_1,b_2}(\varepsilon)} \right\} , \label{eq:oracle}
\end{equation}
with probability larger than $1 - |\Lambda| \left( e^{-(\sqrt{b_1} - 1)^2} + e^{N (\log(2)-b_2) + \log(n)} \right)$.
\end{thm}

\begin{rmq}
For a fixed set $\Lambda$, we asymptotically (in $n$ and $p$) have that
$$V_{n,p}:= V_{b_1,b_2}(\varepsilon) \sim \frac{1}{n}+\sqrt{\frac{\log(n)}{p}}$$
meaning that for a well chosen $p$ depending on $n$, $V_{n,p}$ tends to zero when both $n,p$ tend to infinity. Therefore, when $n\to +\infty$, the probability that $\vert \hat{\br}_{n,p}^{\hat{\varepsilon}}-r^{0}\vert$ is close to zero is as arbitrarily small as $b_1$ and $b_2$ are chosen sufficiently large while satisfying \eqref{def:b1b2}.
\end{rmq}

\begin{proof}
We first start as in the proof of Proposition 1 in \cite{lacour2016minimal} by showing that for any $\varepsilon\in\Lambda$
\begin{equation}
\vert \hat{\br}_{n,p}^{\hat{\varepsilon}}-r^{0}\vert \leq \sqrt{2 B(\varepsilon) + 6 V(\varepsilon)}  + \vert \hat{\br}_{n,p}^{\varepsilon}-r^{\varepsilon}\vert + D(\varepsilon). \label{eq:LacourMassart}
\end{equation}
For completeness, we repeat the arguments in \cite{lacour2016minimal} yielding to inequality \eqref{eq:LacourMassart} as they allow to shed some lights on the basic principles of the GL method. For any fixed $\varepsilon\in\Lambda$  one has that
\begin{eqnarray}
\vert \hat{\br}_{n,p}^{\hat{\varepsilon}}-r^{0}\vert & \leq & \vert \hat{\br}_{n,p}^{\hat{\varepsilon}}-\hat{\br}_{n,p}^{\varepsilon} \vert + \vert \hat{\br}_{n,p}^{\varepsilon}-r^{0}\vert  \leq   \vert \hat{\br}_{n,p}^{\hat{\varepsilon}}-\hat{\br}_{n,p}^{\varepsilon} \vert +  \vert \hat{\br}_{n,p}^{\varepsilon}-r^{\varepsilon}\vert + \vert r^{\varepsilon}-r^{0}\vert \nonumber \\
& \leq &   \vert \hat{\br}_{n,p}^{\hat{\varepsilon}}-\hat{\br}_{n,p}^{\varepsilon} \vert +  \vert \hat{\br}_{n,p}^{\varepsilon}-r^{\varepsilon}\vert + D(\varepsilon).  \label{eq:LacourMassar1}
\end{eqnarray}
Then, for any $\tilde{\varepsilon}\leq\varepsilon$, the definition \eqref{eq:B_tradeoff} of the bias term implies that $\vert \hat{\br}_{n,p}^{\tilde{\varepsilon}}-\hat{\br}_{n,p}^{\varepsilon} \vert^2 \leq B(\varepsilon) +  3V(\tilde{\varepsilon}) $ which can also be written  as
$
\vert \hat{\br}_{n,p}^{\tilde{\varepsilon}}-\hat{\br}_{n,p}^{\varepsilon} \vert^2 \leq B(\max(\varepsilon,\tilde{\varepsilon})) +  3V(\min(\varepsilon,\tilde{\varepsilon})) 
$
for all $\varepsilon ,  \tilde{\varepsilon} \in \Lambda$. Therefore, by definition \eqref{eq:GL_method} of $\hat{\varepsilon}$, one obtains that
\begin{equation}
\vert \hat{\br}_{n,p}^{\hat{\varepsilon}}-\hat{\br}_{n,p}^{\varepsilon} \vert^2 \leq B(\max(\varepsilon,\hat{\varepsilon})) + 3 V(\min(\varepsilon,\hat{\varepsilon})) \leq B(\varepsilon) +  3V(\varepsilon) + \max(B(\varepsilon) , 3V(\varepsilon)). \label{eq:LacourMassar2}
\end{equation}
Hence, inserting inequality \eqref{eq:LacourMassar2} into \eqref{eq:LacourMassar1} finally yields inequality \eqref{eq:LacourMassart}.

Now that inequality \eqref{eq:LacourMassart} has been established,  the main steps in the proof are the control  of the stochastic terms $B(\varepsilon)$ and $\vert \hat{\br}_{n,p}^{\varepsilon}-r^{\varepsilon}\vert$. First, using the triangle inequality
$$
\vert \hat{\br}_{n,p}^{\varepsilon}-r^{\varepsilon}\vert \leq \vert \hat{\br}_{n}^{\varepsilon}-r^{\varepsilon}\vert  + \vert \hat{\br}_{n}^{\varepsilon} - \hat{\br}_{n,p}^{\varepsilon}\vert 
$$
we obtain by Proposition \ref{prop:concentration1} and Proposition \ref{prop:concentration2} that are presented in the appendix, that for any $u_1 > 0$ and $u_2 > 0$,
 $$
\P \left( \vert \hat{\br}_{n,p}^{\varepsilon}-r^{\varepsilon}\vert > (1+u_1)  \frac{2 \sqrt{2} L_{\rho,\varepsilon}}{\varepsilon \sqrt{n}} +  \sqrt{ \frac{2L_{\rho,\varepsilon}}{\varepsilon}\left(u_2 + \rho( N + \sqrt{N}) \right)} \right) \leq \exp\left(  -u_1^2 \right) + 2^N n \exp\left(  - p u_2^2 \right)
$$
where $p=\min_{1 \leq i \leq n} p_i$. Hence, choosing $u_1 > \sqrt{b_1} - 1$ and $u_2 >  \sqrt{b_2 \frac{N}{p}}$, implies that
\begin{equation}
\vert \hat{\br}_{n,p}^{\varepsilon}-r^{\varepsilon}\vert \leq \sqrt{V(\varepsilon)}, \quad \mbox{for all} \quad \varepsilon \in \Lambda, \label{eq:boundVarProba}
\end{equation}
with probability larger than $1 - |\Lambda| \left( e^{-(\sqrt{b_1} - 1)^2} + e^{N (\log(2)-b_2) + \log(n)} \right)$. To control $B(\varepsilon)$ with a bias term, we use the upper bound
$$
\vert \hat{\br}_{n,p}^{\varepsilon}-\hat{\br}_{n,p}^{\tilde{\varepsilon}}\vert^2 \leq \left( \vert \hat{\br}_{n,p}^{\varepsilon}-r^{\varepsilon}\vert + \vert  \hat{\br}_{n,p}^{\tilde{\varepsilon}} - r^{\tilde{\varepsilon}}\vert + \vert r^{\varepsilon}-r^{\tilde{\varepsilon}} \vert \right)^2
$$
combined with inequality \eqref{eq:boundVarProba} to obtain that
$$
\vert \hat{\br}_{n,p}^{\varepsilon}-\hat{\br}_{n,p}^{\tilde{\varepsilon}}\vert^2 \leq \left( \sqrt{V(\varepsilon)} + \sqrt{V(\tilde{\varepsilon})} + \vert r^{\varepsilon}-r^{\tilde{\varepsilon}} \vert \right)^2 \leq 3 \left(V(\varepsilon) +  V(\tilde{\varepsilon}) + \vert r^{\varepsilon}-r^{\tilde{\varepsilon}} \vert^2\right),
$$
for all $\varepsilon$ and $\tilde{\varepsilon}$ belonging to $\Lambda$, with probability $1 - |\Lambda| \left( e^{-(\sqrt{b_1} - 1)^2} + e^{N (\log(2)-b_2) + \log(n)} \right)$, which finally implies that (with the same probability) 
\begin{equation}
B(\varepsilon) \leq 3 V(\varepsilon)  + 3 D^2(\varepsilon). \label{eq:boundBeps}
\end{equation}
Combining inequalities \eqref{eq:LacourMassart}, \eqref{eq:boundVarProba} and  \eqref{eq:boundBeps}, we obtain that
\begin{eqnarray*}
\vert \hat{\br}_{n,p}^{\hat{\varepsilon}}-r^{0}\vert & \leq & \sqrt{6 D^2(\varepsilon) + 12 V(\varepsilon)}  + \sqrt{V(\varepsilon)} + D(\varepsilon), \\
& \leq &  \quad (1 + 2 \sqrt{3} ) \left(D(\varepsilon) +  \sqrt{V(\varepsilon)} \right), \quad \mbox{for all} \quad \varepsilon \in \Lambda,
\end{eqnarray*}
with probability larger than $1 - |\Lambda| \left( e^{-(\sqrt{b_1} - 1)^2} + e^{N (\log(2)-b_2) + \log(n)} \right)$, which completes the proof of Theorem  \ref{theo:oracle}.
\end{proof}

\section{Numerical experiments}\label{sec:expe}

In this section, we first illustrate with one-dimensional datasets the performances of the Goldenshluger-Lepski method to choose the regularization parameters $\gamma$ and $\varepsilon$. Then, we report the results from numerical experiments on simulated Gaussian mixtures and flow cytometry dataset in $\R^2$.

Unfortunately, in numerical experiments, we have found that the Lipschitz constant \eqref{eq:defL} leads to a too rough estimate  of the variance term $\E(\vert r^{\varepsilon}-\hat{r}^{\varepsilon}_{n,p}\vert^2)$ which has a magnitude much smaller than the upper bound \eqref{eq:Sinkhorn_rate}. Using the conditions of Theorem \ref{theo:oracle} on $b_1$ and $b_2$ leads  to a quantity $V(\varepsilon)$ which is  much larger than the bias term $B(\varepsilon)$ leading to always choosing the smallest value of $\varepsilon$ in $\Lambda$. To overcome this problem, it is necessary to either scale  the magnitude of $V(\varepsilon)$ by choosing small values of $b_1$ and $b_2$, 
or to improve the upper bound \eqref{eq:Sinkhorn_rate} using numerical methods. To this end, we use Monte-Carlo simulations to obtain the right order for the variance term, which allows to show that the Goldenshluger-Lepski principle leads to satisfactory choices of $\varepsilon$ in this setting.

\subsection{Simulated data: one-dimensional Gaussian mixtures}

We illustrate GL's principle for the one-dimensional example of random Gaussian mixtures that is displayed in Figure \ref{fig:ex_gaussian1D}. Our dataset consists of observations $(X_{i,j})_{1\leq i\leq n\ ; 1\leq p}$  sampled from  $n=15$ random distributions $\bnu_i$ that are mixtures of two Gaussian distributions with weights $(0.35,0.65)$, random means respectively belonging  to  the intervals $[-6,-2]$ and $[2,6]$ and random variances  both belonging to the interval $(0,2]$. For each random mixture distribution, we sample  $p=50$ observations. A first step is to perform Monte-Carlo experiments over $20$ regularized barycenters $\hat{r}^{\varepsilon}_{n,p}$ for each different values of $\varepsilon$ in the interval $[0.2,5]$. The results are presented in Figure \ref{fig:monte_carlo_gaussian1D} where we plot the estimation of the variance $\E(\vert r^{\varepsilon}-\hat{r}^{\varepsilon}_{n,p}\vert^2)$ with respect to the regularization parameter $\varepsilon\in [0.2,5]$ for different values of $N=2^6,2^7, 2^8$.
A first aspect is that the 
variance term estimated by Monte Carlo simultions
decreases as the regularization parameter increases.



\begin{figure}
\centering
\includegraphics[width=0.95 \textwidth,height=0.55\textwidth]{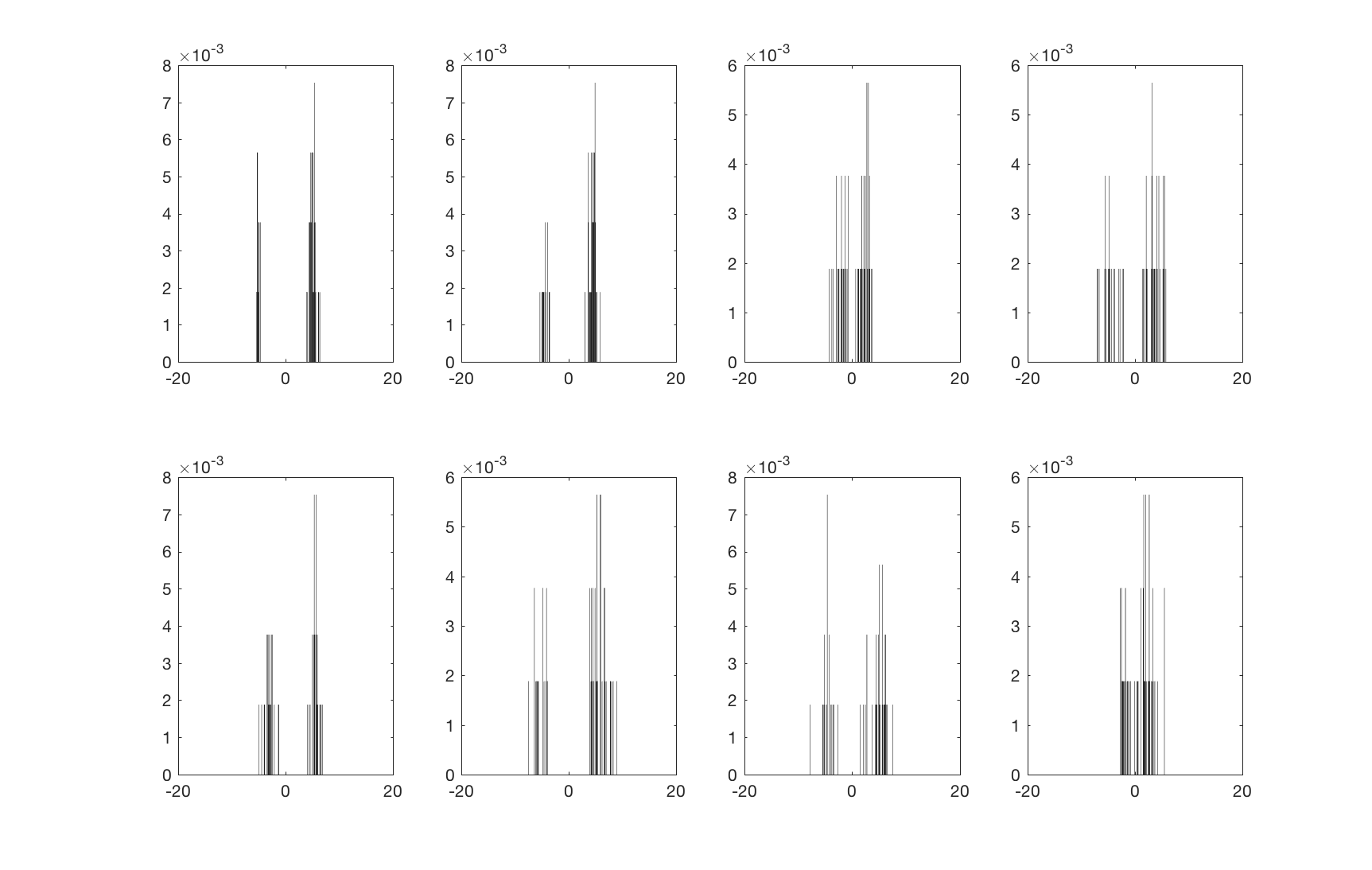} 
\caption{A subset of $8$ histograms (out of $n=15$) obtained with random variables sampled from  one-dimensional Gaussian mixtures distributions $\bnu_i$ (with random means and variances). Histograms are constructed by binning the data  $(X_{i,j})_{1\leq i\leq n\ ; 1\leq p}$ on a grid $\Omega_N$ of size $N=2^8$.}
\label{fig:ex_gaussian1D}
\end{figure}

\begin{figure}
\centering
\hspace{-1cm}\begin{tabular}{ccc}
\includegraphics[width=0.35 \textwidth,height=0.25\textwidth]{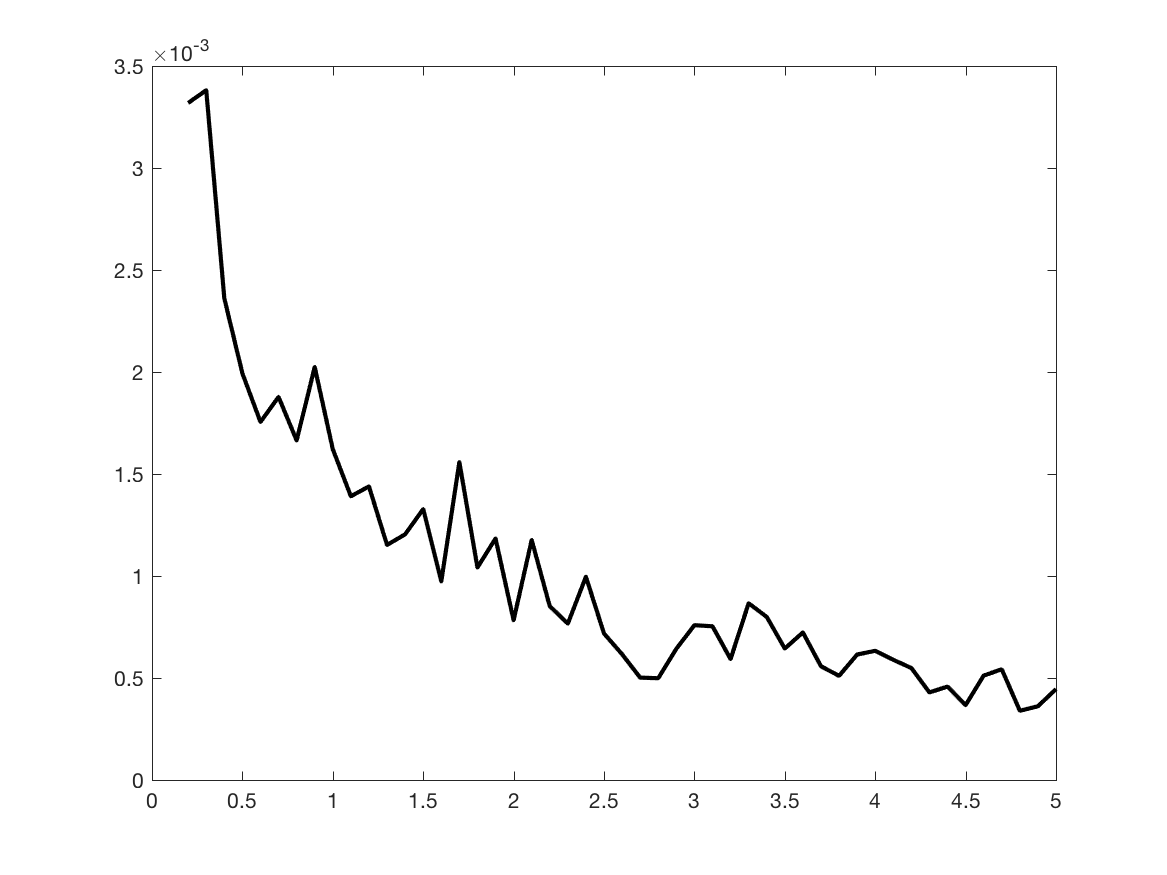} &\hspace{-1cm}
\includegraphics[width=0.35 \textwidth,height=0.25\textwidth]{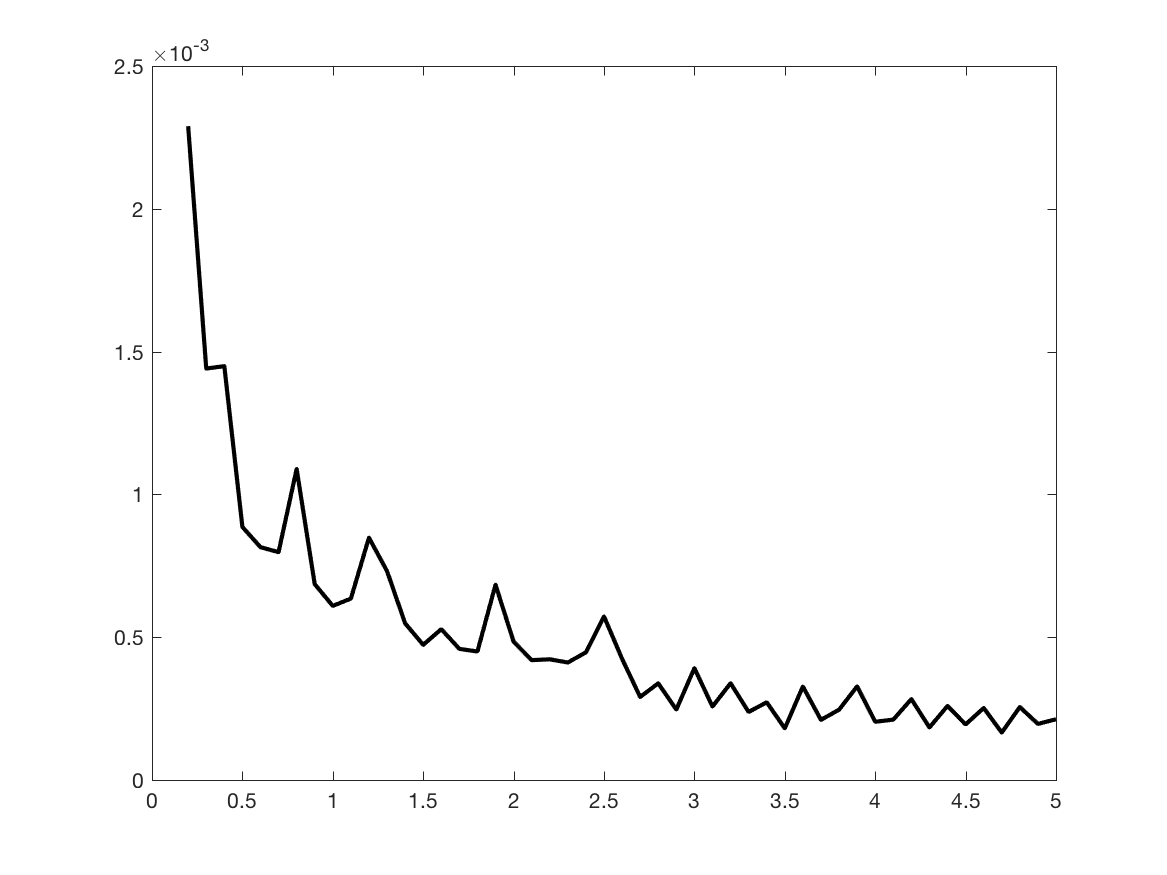} &\hspace{-1cm}
\includegraphics[width=0.35 \textwidth,height=0.25\textwidth]{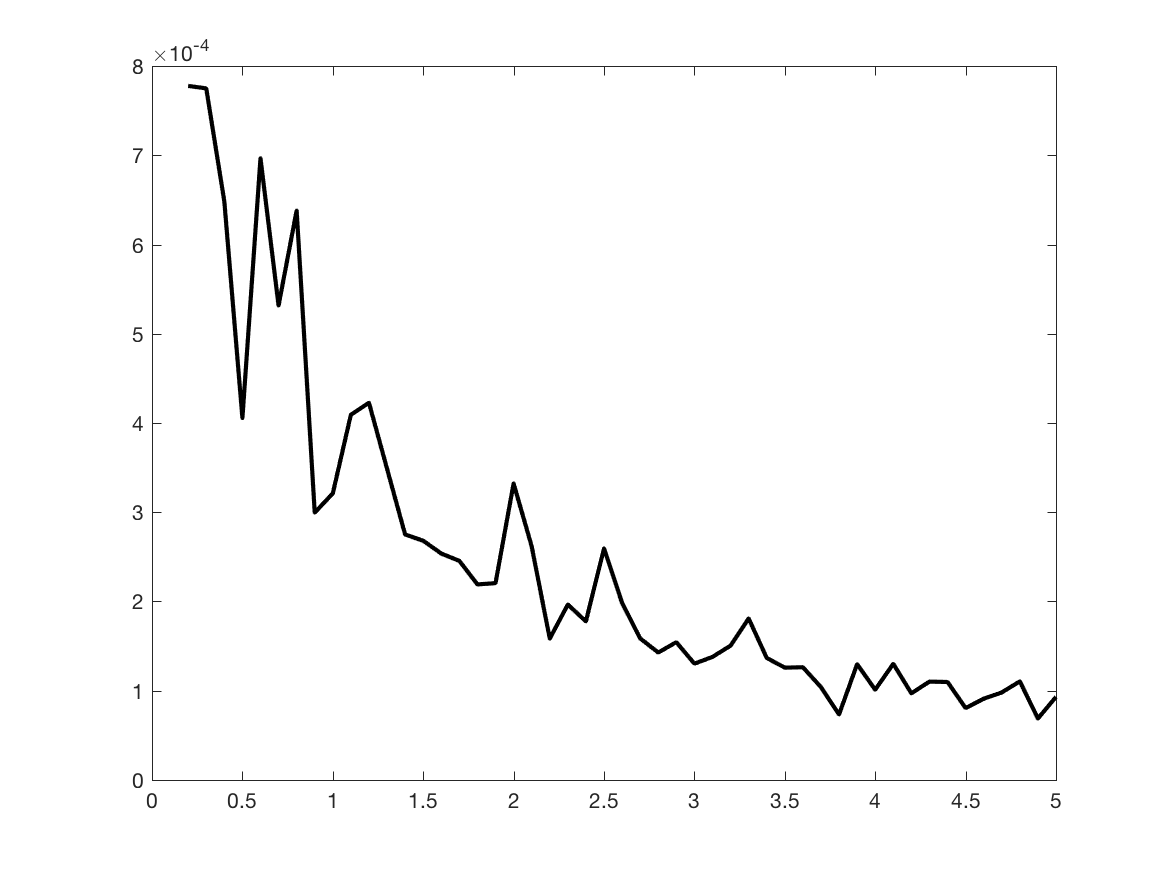} \\
(a) $N = 2^6$ &\hspace{-1cm} (b) $N = 2^7$ &\hspace{-1cm} (c) $N = 2^8$
\end{tabular}
\caption{A Monte-Carlo experiments for estimating the variance $\E(\vert r^{\varepsilon}-\hat{r}^{\varepsilon}_{n,p}\vert^2)$ for three values of the grid $N$.}
\label{fig:monte_carlo_gaussian1D}
\end{figure}

From now on, we  fix the size $N=2^8$ of the grid, and we comment on the choice of the parameters  $\hat{\varepsilon}$ and $\hat{\gamma}$. To obtain this data-driven choice of regularization,   we use the Goldenshluger-Lepski principle where the variance term $ V(\varepsilon)$ is given by the variance function displayed in Figure \ref{fig:monte_carlo_gaussian1D} obtained via Monte-Carlo simulations.
We display in Figure \ref{fig:Sinkhorn_gaussian1D}(a) the trade-off function $B(\varepsilon)+3V(\varepsilon)$, and we discuss the influence of  $\varepsilon$ on the smoothness the Sinkhorn barycenter.

From Figure \ref{fig:Sinkhorn_gaussian1D}(b), we observe that,  when the parameter $\varepsilon=0.2$ is small (dotted blue curve), then  the corresponding Sinkhorn barycenter $\hat{\br}_{n,p}^{\varepsilon}$  is irregular, and it presents spurious peaks. On the contrary, too much regularization, e.g.\ $\varepsilon=5$, implies that the barycenter (dashed green curve) is flattened and its mass is spread out. The optimal barycenter (solid red curve), that is $\hat{\br}_{n,p}^{\hat{\varepsilon}}$ for $\hat{\varepsilon}=3.8$ minimizing the  trade-off  function \eqref{eq:GL_method},  gives here a good compromise between under and over-smoothing.

\begin{figure}
\centering
\subfigure[]{\includegraphics[width=0.45 \textwidth,height=0.45\textwidth]{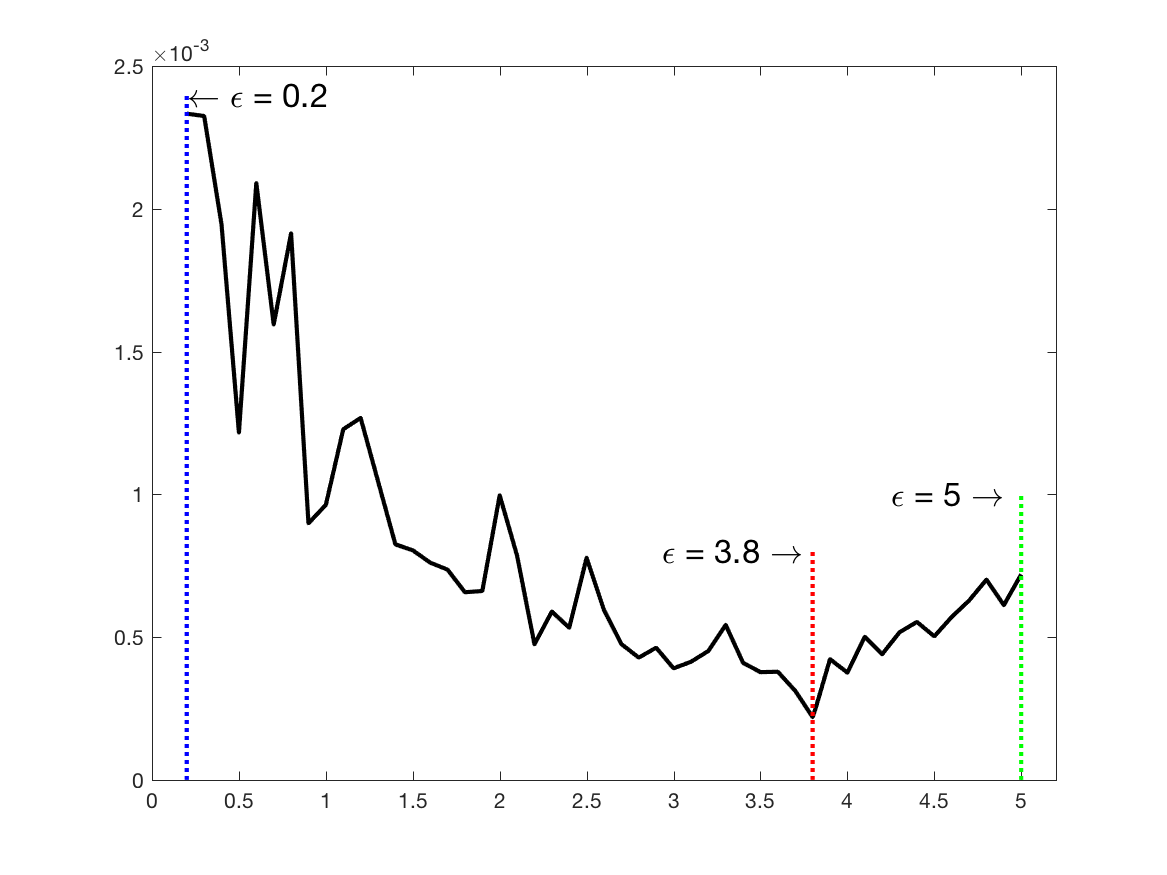}}
\subfigure[]{\includegraphics[width=0.45 \textwidth,height=0.45\textwidth]{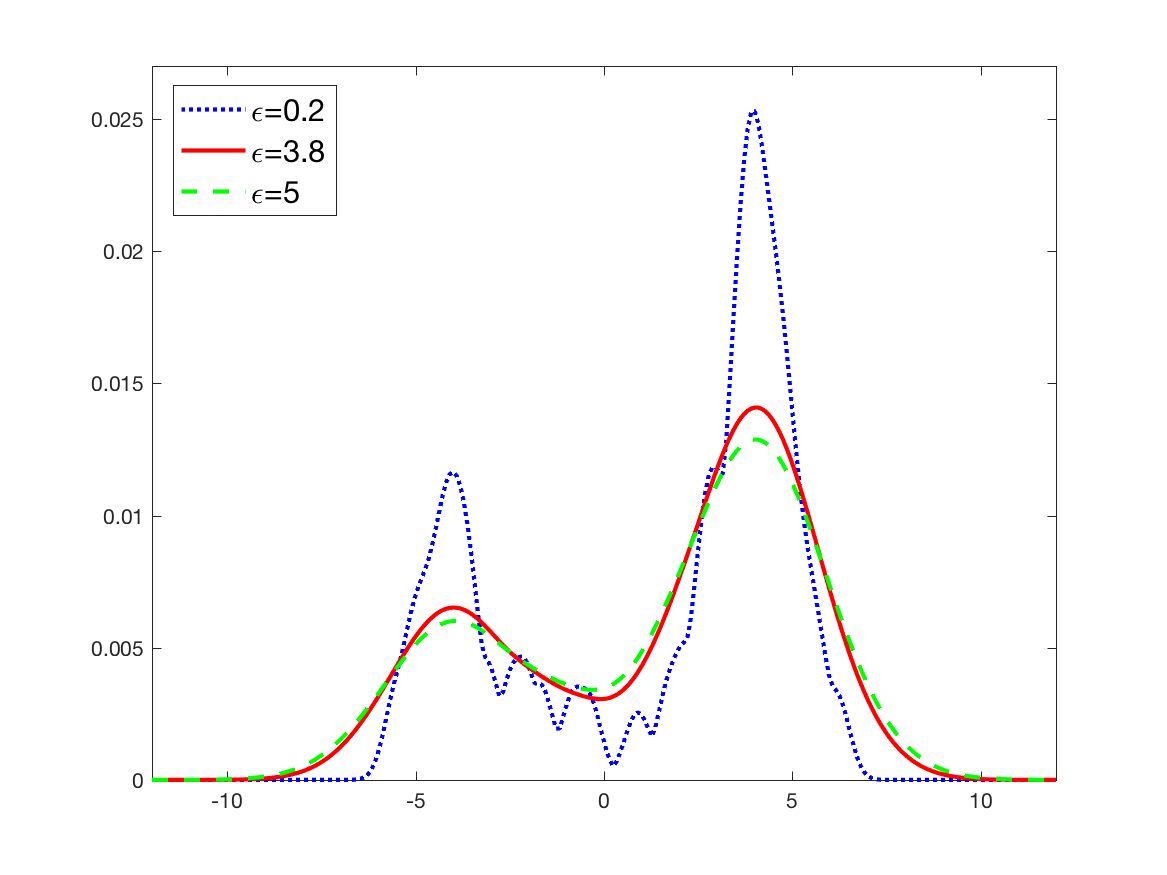}}
\caption{One dimensional Gaussian mixtures dataset and Sinkhorn barycenters. (a) The trade-off function $\varepsilon\mapsto \ B(\varepsilon)+3V(\varepsilon)$ which attains its optimum at $\hat{\varepsilon}=3.8$. (b) Three Sinkhorn barycenters $\hat{\br}_{n,p}^{\hat{\varepsilon}}$ associated to the parameters $\varepsilon=0.2, 3.8, 5$.}
\label{fig:Sinkhorn_gaussian1D}
\end{figure}

We repeat the same experiment for the penalized barycenter $\hat{\bfun}_{n,p}^{\gamma}$ of Section \ref{sec:Regbar} with the Sobolev norm $H^1(\Omega)$  in the penalization function $E$ \eqref{ex_penalty} . In particular, we remark that the size of the grid does not appear explicitly in the the upper bound of the variance function in inequality \eqref{eq:rate_Regbar}. The Monte-Carlo experiments are presented in Figure \ref{fig:monte_carlo_gaussian1D_pen} for $N=2^8$. This gives us an approximation for the variance term in Goldenshluger-Lepski.

\begin{figure}
\centering
\includegraphics[width=0.45 \textwidth,height=0.35\textwidth]{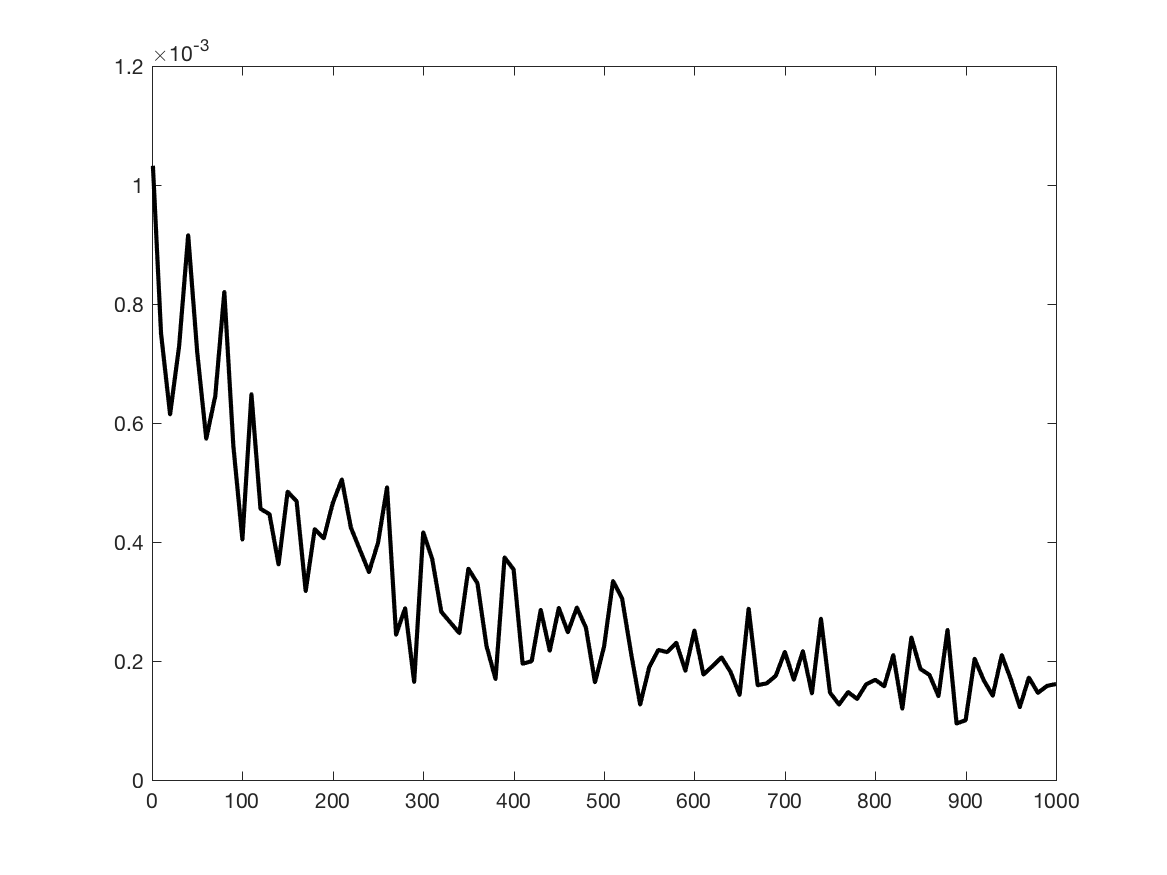}
\caption{A Monte-Carlo experiments for estimating the variance $\E\left(\Vert \hat{\bfun}_{n,p}^{\gamma}-f_{\P}^{\gamma}\Vert^2_{\mathbb{L}_2(\Omega)}\right)$.}
\label{fig:monte_carlo_gaussian1D_pen}
\end{figure}

The results of the Goldenshluger-Lepski technique are displayed in Figure \ref{fig:Regbar_gaussian1D}. The advantage of choosing a Sobolev penalty function over an entropy term is that the mass of the barycenter is overall less spread out and the spikes are  sharper. However, for a small regularization parameter $\gamma=20$ (dotted blue curve), the barycenter $\bfun_{n,p}^{\gamma}$ presents a lot of irregularities as the penalty function tries to minimize its $\L_2$-norm. When the regularization parameter increases in a significant way ($\gamma=1000$ associated to the dashed green curve), the irregularities disappear and the support of the penalized barycenter becomes wider. The GL's principle leads to the choice  $\hat{\gamma}=840$ which corresponds to a penalized barycenter (solid red curve) that is satisfactory.

\begin{figure}
\centering
\subfigure[]{\includegraphics[width=0.45 \textwidth,height=0.45\textwidth]{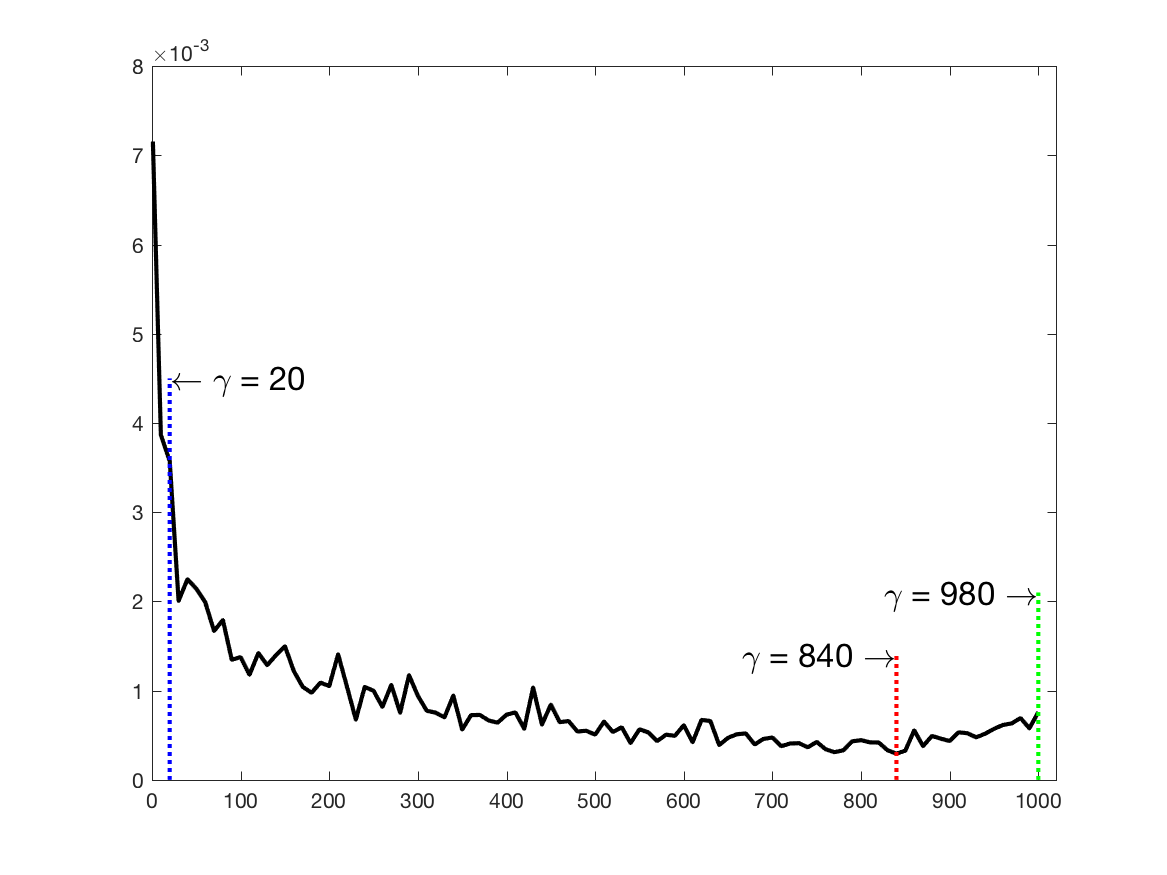}}
\subfigure[]{\includegraphics[width=0.45 \textwidth,height=0.45\textwidth]{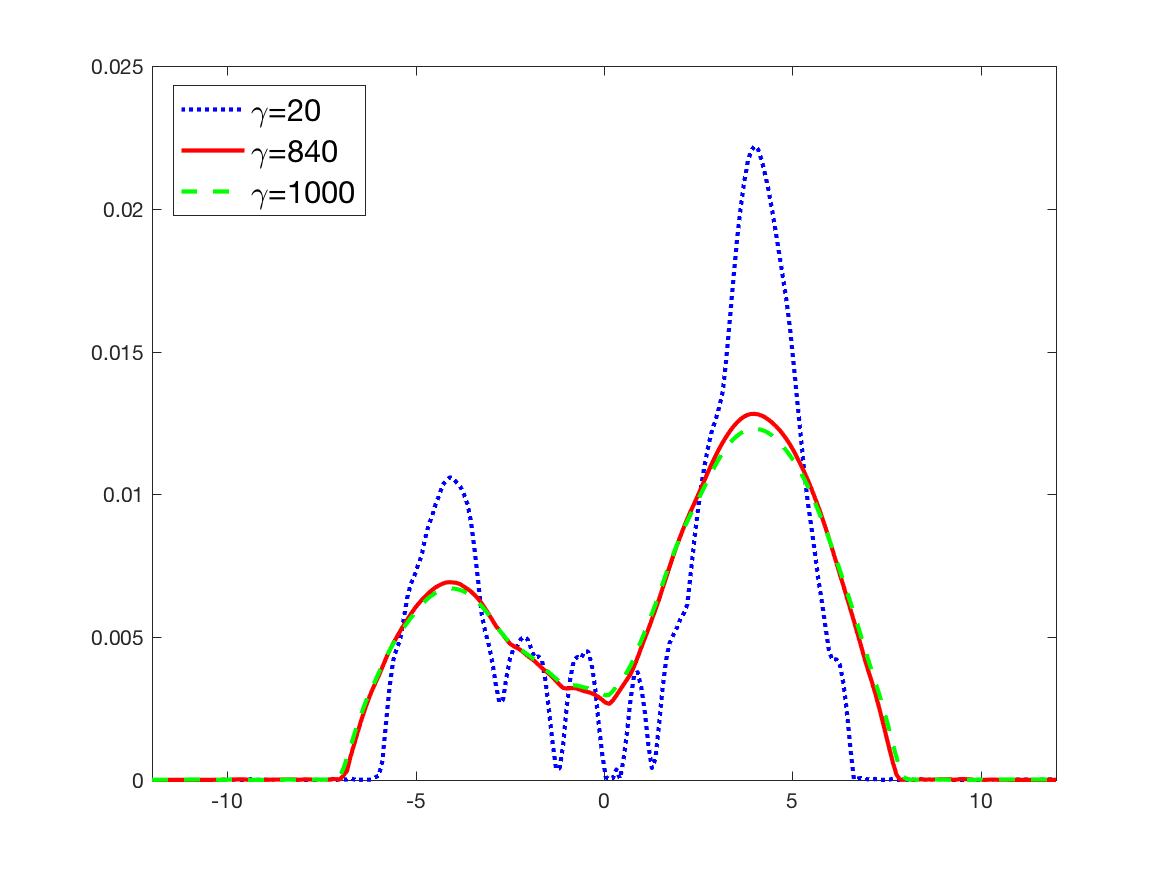}}
\caption{One dimensional Gaussian mixtures dataset and penalized barycenters. (a) The trade-off function $\gamma\mapsto \ B(\gamma)+3V(\gamma)$ which attains its optimum at $\hat{\gamma}=840$. (b) Three penalized barycenters $\bfun_{n,p}^{\gamma}$ associated to the parameters $\gamma=20, 840, 100$.}
\label{fig:Regbar_gaussian1D}
\end{figure}

We compare these Wasserstein barycenters to the  Euclidean mean $\bar{f}_{n,p}$  \eqref{eq:Eucli_mean},  obtained after  a  pre-smoothing step of the data for each subject using the kernel method. From Figure \ref{fig:mean_gaussian1D}, the density  $\bar{f}_{n,p}$  is very irregular and it suffers from mis-alignment issues. The irregularity of this estimator mainly comes from the low sample size per subject ($p=50$).

\begin{figure}
\centering
\includegraphics[width=0.45 \textwidth,height=0.45\textwidth]{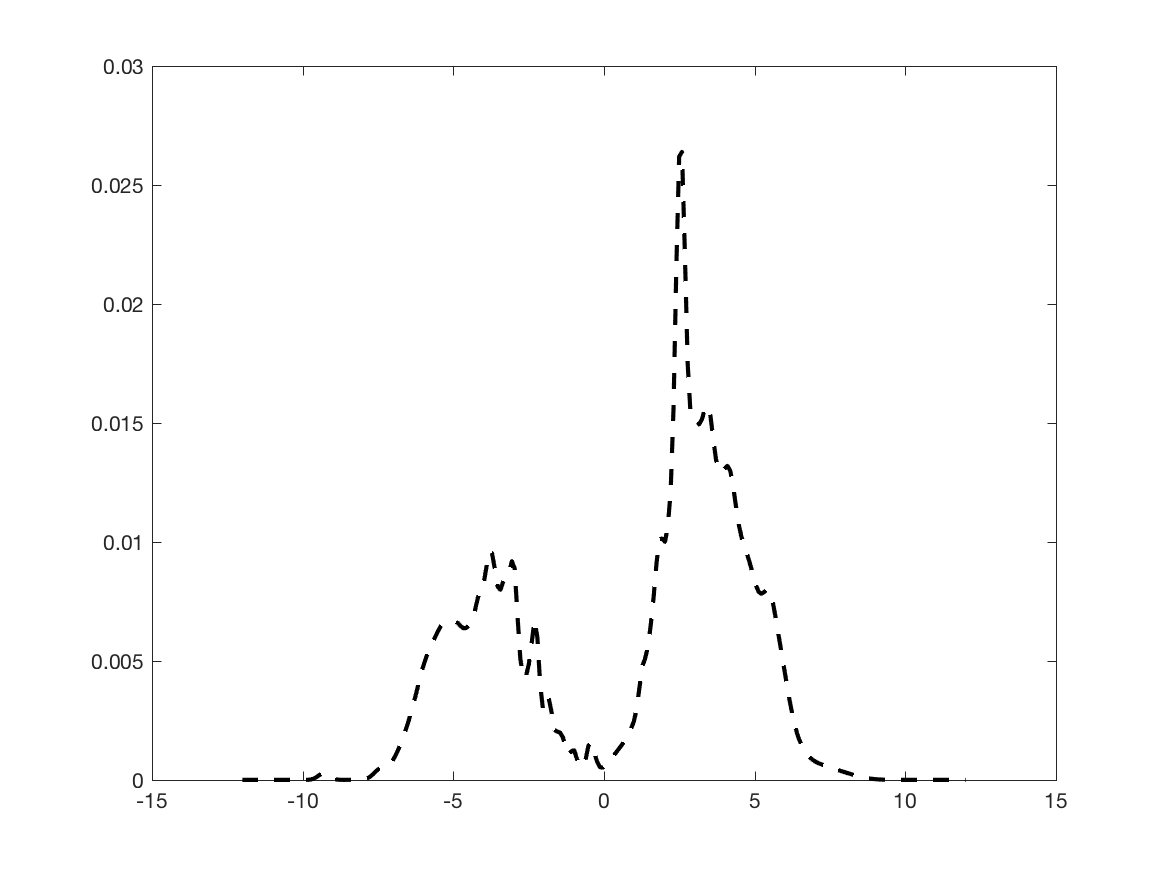}
\caption{ Euclidean mean density $\bar{f}_{n,p}$ of the one dimensional Gaussian mixtures dataset using a preliminary smoothing step of each subject with a Gaussian kernel.}
\label{fig:mean_gaussian1D}
\end{figure}

\subsection{Sinkhorn versus penalized  barycenters}

To conclude these numerical experiments with one-dimensional simulated data, we would like to point out that computing the Sinkhorn barycenter is much faster than computing the penalized barycenter. Indeed, entropy regularization of the transport plan in the Wasserstein distance has been first introduced in order to reduce the computational cost of a transport distance. The computation of an unregularized transport distance requires $\mathcal{O}(N^3\log N)$ operations for discrete probability measures with a support of size $N$ when the computation of a Sinkhorn divergence only takes $\mathcal{O}(N^2)$ operations at each iteration of a gradient descent (see e.g.\ \cite{cuturi2013sinkhorn}). We have also found that the Sinkhorn barycenter yields more satisfying estimators in terms of smoothness. Therefore, in the rest of this section, we do not consider the penalized barycenter anymore.

\subsection{Simulated data: two-dimensional Gaussian mixtures}
\label{sec:simu_2Dgaussian}
In this section, we illustrate  our methods for two-dimensional data. We consider a simulated example of observations $(X_{i,j})_{1\leq i\leq n\ ; 1\leq j\leq p}$  sampled from  $n=15$ random distributions $\bnu_i$ that are a mixture of three multivariate Gaussian distributions $\bnu_i =\sum_{j=1}^3\theta_j\mathcal{N}(\bm_j^i,\bGamma_j^i)$ with fixed weights $\theta=(1/6,1/3,1/2)$. The means $\bm_j^i$ and covariance matrices $\bGamma_j^i$ are random variables with expectation given by (for $j=1,2,3$)
$$m_1=\left(\begin{array}{c} 0 \\ 0 \end{array}\right), \quad m_2=\left(\begin{array}{c} 7 \\ 4 \end{array}\right), \quad m_3=\left(\begin{array}{c} 1 \\ 9 \end{array}\right), \quad \mbox{and} \quad \Gamma_{1}  = \Gamma_{2} = \Gamma_{3} =\left(\begin{array}{cc} 1 & 0\\ 0 & 1 \end{array}\right).$$
The covariances $\Gamma_i$ are chosen diagonal to ensure that the perturbed random covariances that form our dataset are positive semi-definite matrices to properly define the Gaussian distributions associated. 
For each $i=1,\ldots,n$, we simulate a sequence  $(X_{i,j})_{1\leq j\leq p}$ of $p=50$ iid random variables sampled from $\bnu_i  =\sum_{j=1}^3\theta_j\mathcal{N}(\bm_j^i,\bGamma_j^i)$ where $\bm_j^i$ (resp. $\bGamma_j^i$) are random vectors (resp. matrices) such that each of their coordinate follows a uniform law centered in $m_j$ with amplitude $\pm 2$ (resp. each of their diagonal elements follows a uniform law centered in $\Gamma_j$ with amplitude $\pm 0.95$). We display in Figure \ref{fig:ex_gaussian2D} the dataset $(X_{i,j})_{1\leq j\leq p\ ; 1\leq i\leq n}$. Each $X_{i,j}$ is then binned on a grid of size $64\times 64$ (thus $N = 4096$).

\begin{figure}
\centering
\includegraphics[width=0.95 \textwidth,height=0.55\textwidth]{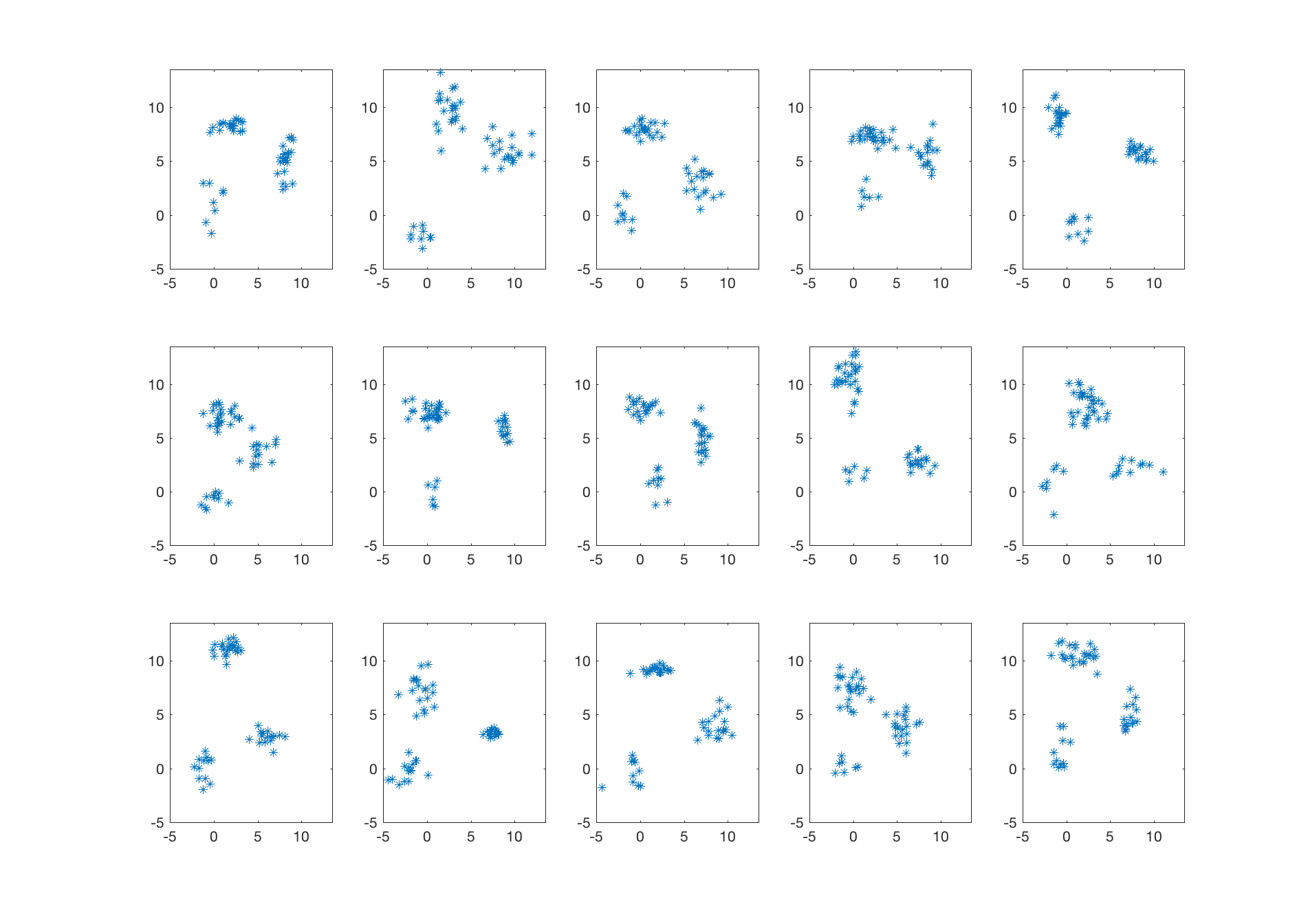} 
\caption{Dataset $(X_{i,j})_{1\leq j\leq p\ ; 1\leq i\leq n}$ generated from $n=15$ two-dimensional random Gaussian mixtures $\bnu_i$ with $p=50$.}
\label{fig:ex_gaussian2D}
\end{figure}

First, we compute $60$ Sinkhorn barycenters by letting $\varepsilon$ ranging from $0.1$ to $6$. We draw in Figure \ref{fig:Sinkhorn_gaussian2D}(a) the trade-off function that is also based on $10$ Monte-Carlo experiments to approximate the term $V$, with its minimizer at $\hat{\varepsilon}=3$. The corresponding Sinkhorn barycenter $\hat{\br}_{n,p}^{\hat{\varepsilon}}$  is displayed in Figure \ref{fig:Sinkhorn_gaussian2D}(b). We also present the Euclidean mean $\bar{f}_{n,p}$ (after a preliminary smoothing step) in Figure \ref{fig:mean_gaussian2D}(b). The Sinkhorn barycenter has three distinct modes. Hence, this approach handles in a very efficient way the scaling and translation variations in the dataset (which corresponds to the correction of the mis-alignment issue).  On the other hand, the Euclidean mean mixes the distinct modes of the Gaussian mixtures. It is thus less robust to outliers since the support of the barycenter is significantly spread out.

\begin{figure}
\centering
\subfigure[]{\includegraphics[width=0.45 \textwidth,height=0.45\textwidth]{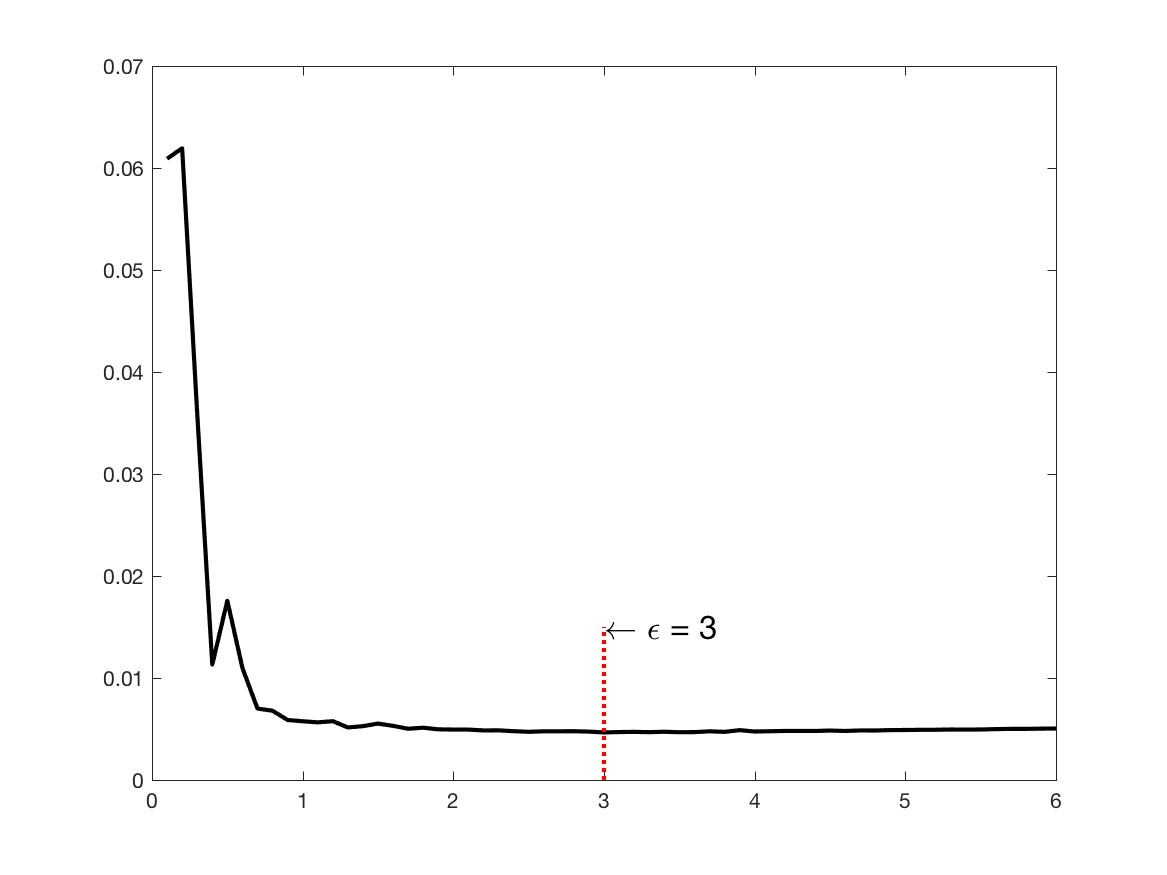}}
\subfigure[]{\includegraphics[width=0.45 \textwidth,height=0.45\textwidth]{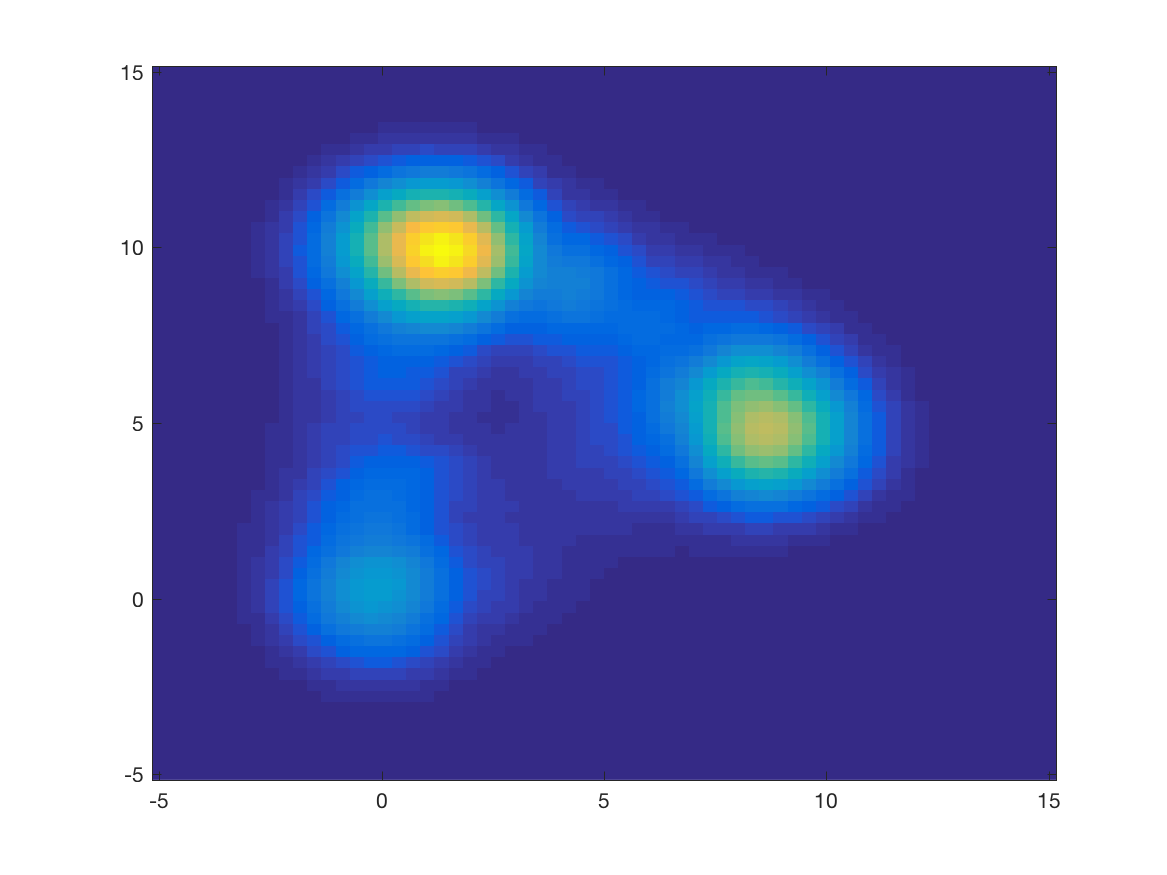}}
\caption{Two-dimensional Gaussian mixtures dataset. (a) The trade-off function $\varepsilon\mapsto \ B(\varepsilon)+3V(\varepsilon)$ which attains its optimum at $\hat{\varepsilon}=3$. (b) The Sinkhorn barycenter $\hat{\br}_{n,p}^{\hat{\varepsilon}}$ for $\hat{\varepsilon}=3$ chosen by the GL's principle.}
\label{fig:Sinkhorn_gaussian2D}
\end{figure}

\begin{figure}
\centering
\includegraphics[width=0.45 \textwidth,height=0.45\textwidth]{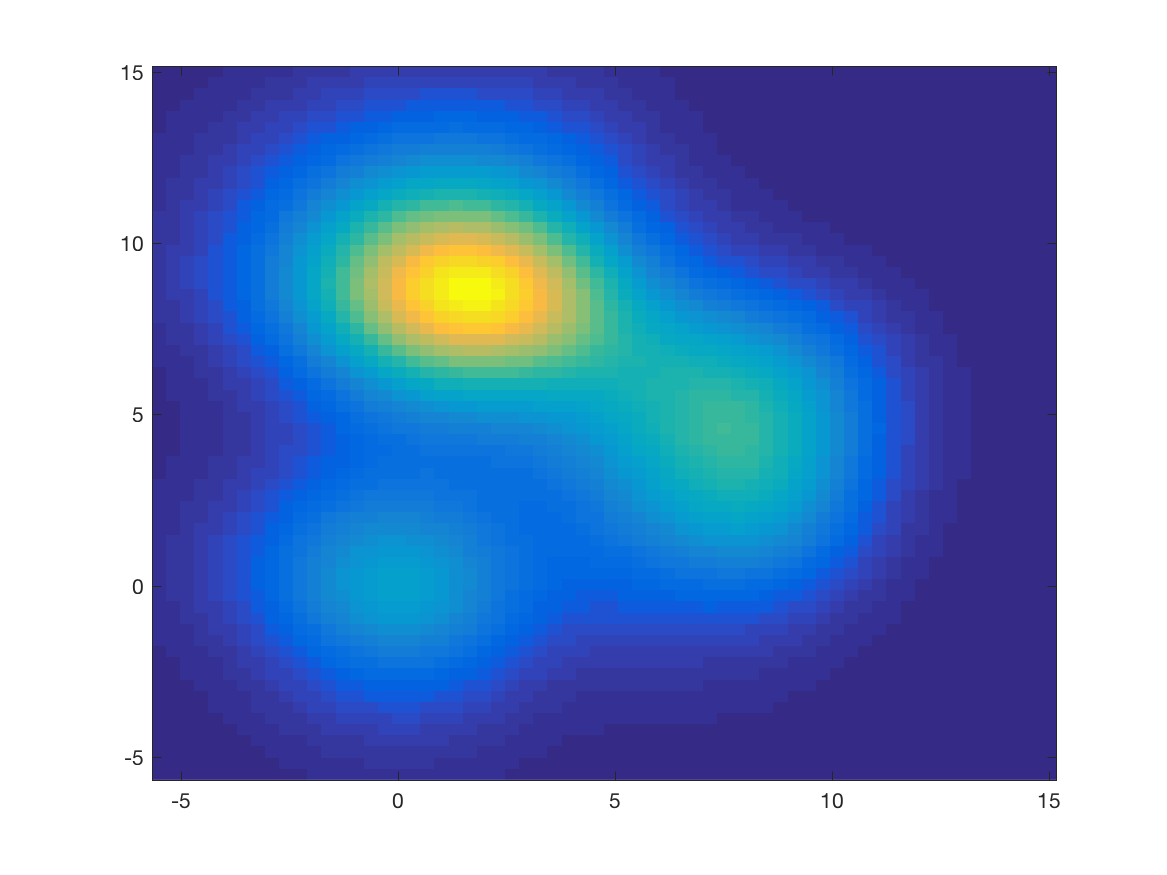}
\caption{Two-dimensional Gaussian mixtures dataset. The Euclidean mean $\bar{f}_{n,p}$ (after a preliminary smoothing step).}
\label{fig:mean_gaussian2D}
\end{figure}

\subsection{Real data: flow cytometry}

We have at our disposal data from flow cytometry that have been described in Section \ref{sec:flow}, and we focus on the FSC and SSC cell markers resulting in the dataset that is displayed  in Figure \ref{fig:ex_cytometry2D}. We again apply a binning of the data on a two-dimensional grid of size $N = 64 \times 64$. In Figure \ref{fig:Sinkhorn_cytometry2D}(a) we plot the trade-off function related to the Sinkhorn barycenters. To that end, we use the results on the Monte-Carlo estimation of the variance in the gaussian case (see the previous Section \ref{sec:simu_2Dgaussian}).
Indeed, as the upper-bound of the variance in \eqref{eq:Sinkhorn_rate} depends explicitly on parameters of the problem, we can compare the parameters from the 2D gaussian case to the parameters of the real cytometry data case. Let us first remark that in both experiments, the size of the grid is chosen as $N=64^2$, the number of distributions is equal to $n=15$, the minimum number of observations per measure is approximately $p\sim 50$, and the collection of parameters $\varepsilon=[0.1,6]$ that we test is the same. However, the grid differs in these two experiments. For the gaussian case, the grid is taken as the square $[-5,15]^2$;  in the cytometry case, the measurements of the cells are within the rectangle $[200,600]\times [0,250]$. Therefore, in the Lipschitz constant term, the quantity $$(\Delta_m)_m := (\adjustlimits\inf_{1 \leq k \leq N} \sup_{1 \leq \ell\leq N} |C_{m\ell} - C_{k \ell}|)_m,$$ which is a vector of length $64$, involving the cost function will differ. For the gaussian case grid we have $\max_m \Delta_m=139.6825$, for the cytometry case grid we have $\max_m \Delta_m=37587$. For the reason, we thus choose to scale the variance term in the GL method by $139.6825/37 587 \simeq 1/269$.
The regularization parameter $\varepsilon$ still ranges from $1$ to $6$. The minimum of the trade-off function is reached  for $\hat{\varepsilon}=2.6$. We  display the corresponding Sinkhorn barycenter  in Figure \ref{fig:Sinkhorn_cytometry2D}(b). This barycenter clearly corrects mis-alignment issues of the data.

\begin{figure}
\centering
\subfigure[]{\includegraphics[width=0.45 \textwidth,height=0.45\textwidth]{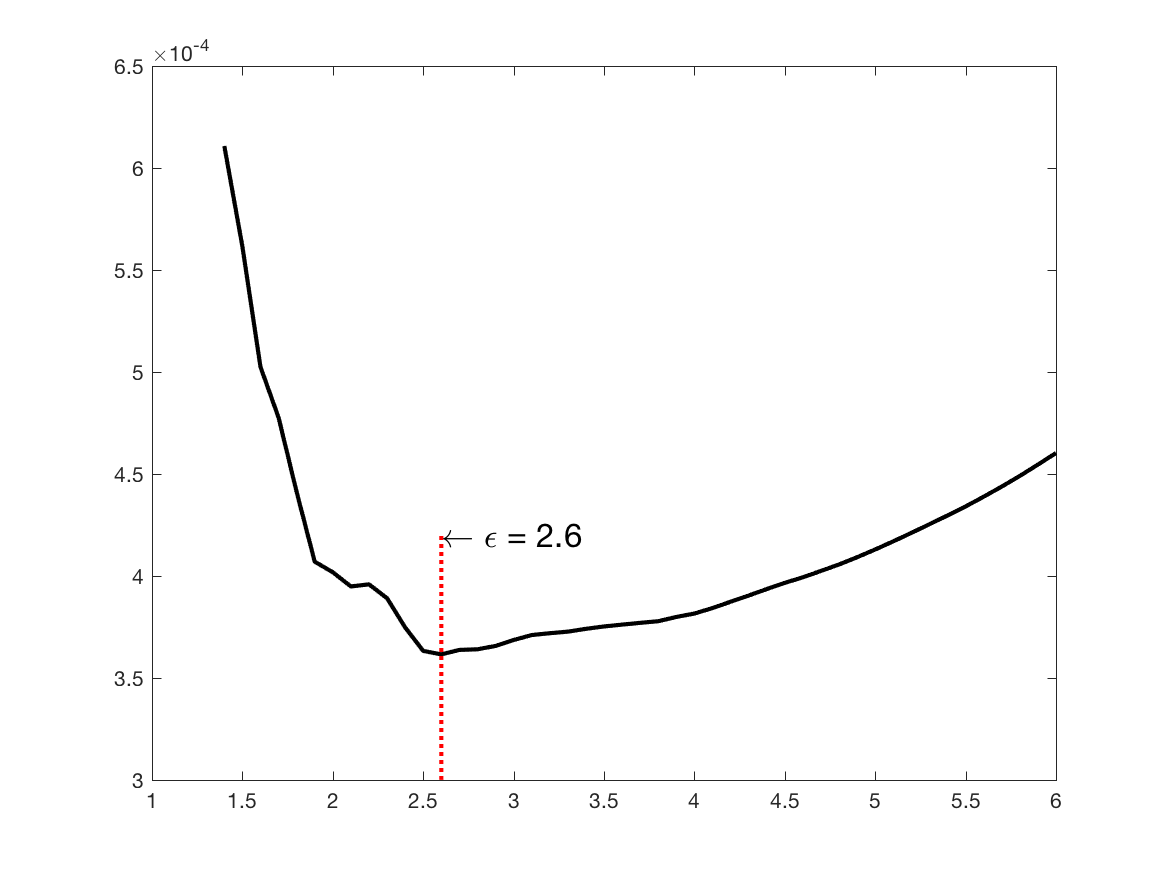}}
\subfigure[]{\includegraphics[width=0.45 \textwidth,height=0.45\textwidth]{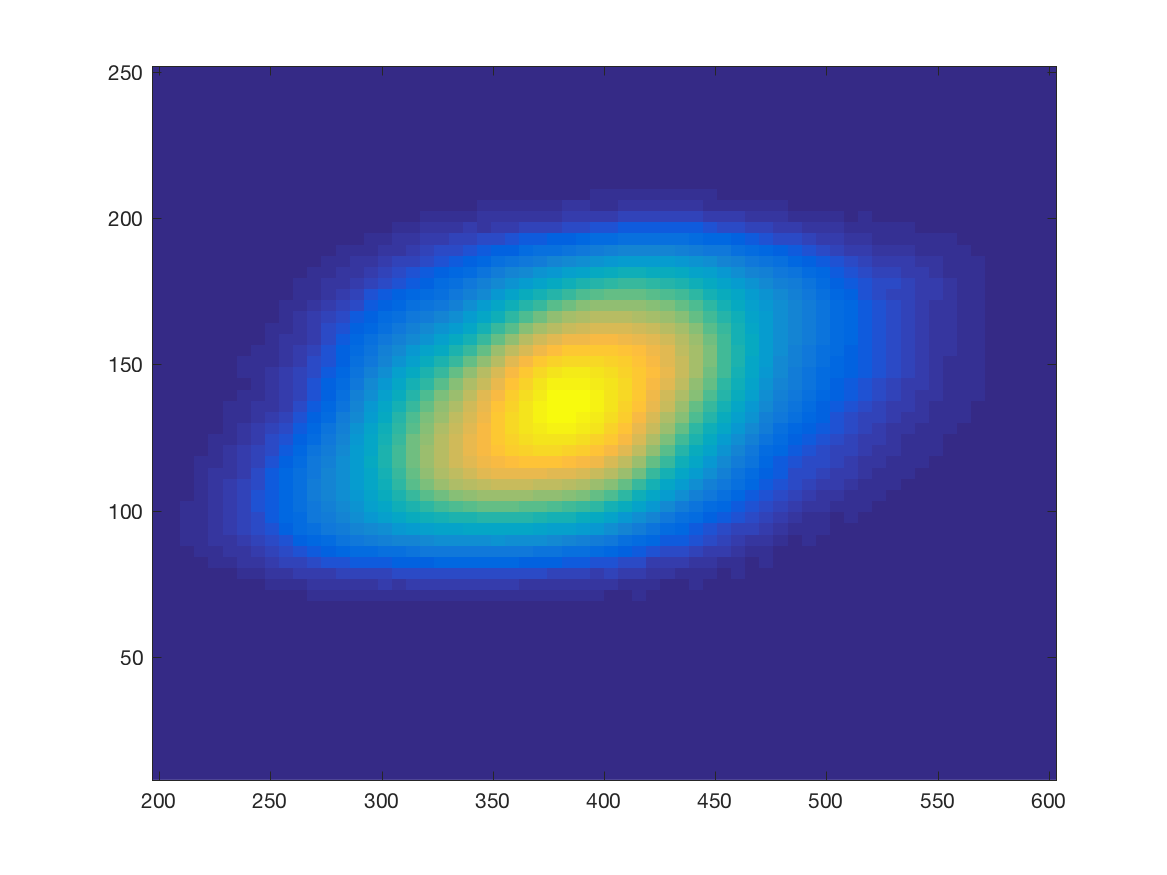}}
\caption{Two dimensional flow cytometry dataset and Sinkhorn barycenter. (a) The trade-off function $\varepsilon\mapsto \ B(\varepsilon)+3V(\varepsilon)$ which attains its optimum at $\hat{\varepsilon}=2.6$. (b) Sinkhorn barycenter $\hat{\br}_{n,p}^{\hat{\varepsilon}}$ associated to the parameter $\hat{\varepsilon}=2.6$.}
\label{fig:Sinkhorn_cytometry2D}
\end{figure}

To analyze the relevance of this result, we present in Figure \ref{fig:mean_cytometry2D}(a) the Euclidean mean $\bar{f}_{n,p}$ of this dataset (after kernel smoothing of the data for each patient). The support of this estimator is again  spread out due to the presence of a strong translation variance  in the data which clearly need to be registered. We also compare our method to the more relevant one proposed in  \cite{hahne2010per} which consists in approximating each of the $15$ subjects with a two dimensional kernel density estimate (with an automatic choice of the bandwidth parameter). The densities obtained are then projected onto a one dimensional space.  Landmarks are estimated by identifying peaks of the resulting one-dimensional densities. Then, these landmarks are registered across the whole dataset in order to finally align the densities. The $\mathbb{L}_2$-mean density obtained after this pre-processing step is displayed in Figure \ref{fig:mean_cytometry2D}(b). This method leads to results that are  similar to the one obtained with a data-driven Sinkhorn barycenter. However, contrary to regularized Wasserstein barycenters that can handle automatically non registered multi-modal densities,  the  procedure  in \cite{hahne2010per} suffers from two main drawbacks: (i) the registration of densities is performed by one-dimensional projections, (ii) the number of peaks to align is chosen manually.
\begin{figure}
\centering
\subfigure[]{\includegraphics[width=0.45 \textwidth,height=0.45\textwidth]{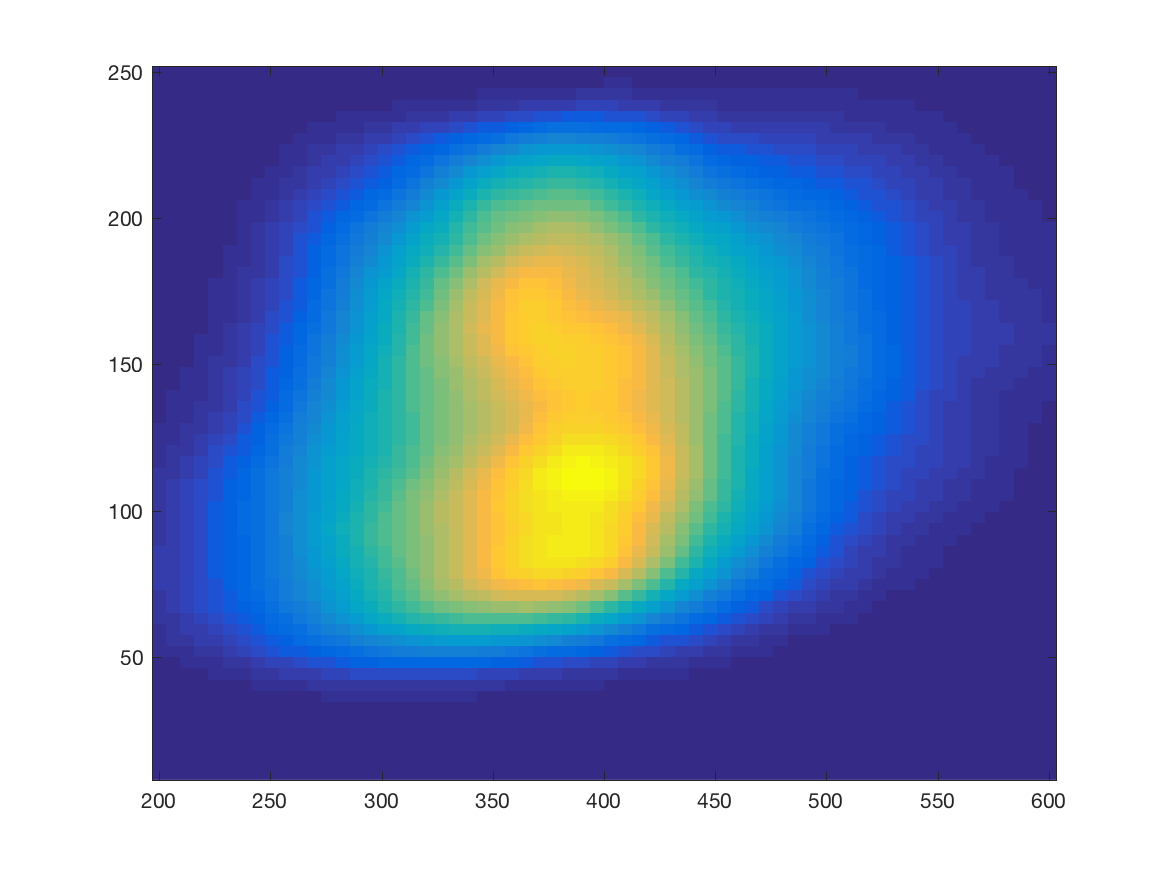}}
\subfigure[]{\includegraphics[width=0.45 \textwidth,height=0.45\textwidth]{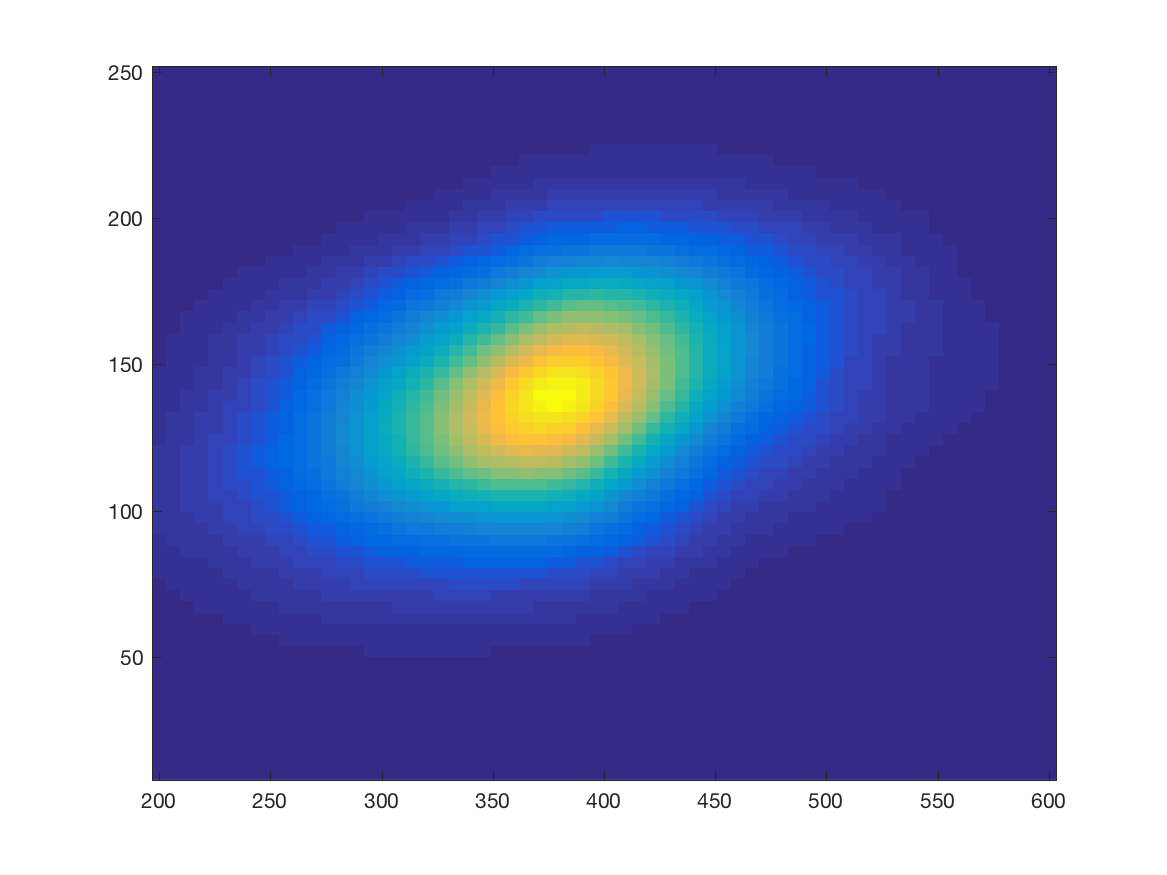}}
\caption{Two dimensional flow cytometry dataset. (a) Euclidean mean $\bar{f}_{n,p}$ of the data (after smoothing but without registration), (b)  $\mathbb{L}_2$-mean  of pre-processed data  using kernel smoothing and density registration by landmark alignment with the method in \cite{hahne2010per}.}
\label{fig:mean_cytometry2D}
\end{figure}
Notice that we have also conducted experiments for Sinkhorn barycenters with  non-equal weights,  corresponding to the proportion of measurements for each  patient. The result being analogous, we do not report them.

\section{Conclusion and perspectives} \label{sec:conclusion}

It would be interesting to derive more accurate upper bounds of the variance term for the Sinkhorn barycenter that would be of a magnitude of the same order than the one found in numerical experiments using Monte Carlo simulations. An additional difficulty is the possibility of using Sinkhorn barycenters with data-driven choice of the regularization parameter for the registration of multiple point clouds beyond the dimension $d \geq 3$, e.g.\ for applications in flow cytometry where the dimension $d$ can be 30 or 40. This is clearly a challenging task.

\paragraph{Acknowlegements.}  This work has been carried out with financial support from the French
State, managed by the French National Research Agency (ANR) in the frame
of the GOTMI project (ANR-16-CE33-0010-01).

\appendix

\section{Strong convexity of the Sinkhorn divergence - Proof of Theorem  \ref{th:strong_convexity_H}} \label{sec:strong_convexity_H}

The proof of Theorem \ref{th:strong_convexity_H} relies on the analysis of the eigenvalues of the Hessian matrix $\nabla^2H_q^{\ast}(g)$ of the functional $H_q^{\ast}$. 

\begin{proposition}
\label{prop:eigenvector_vn}
For all $g\in\R^N$, $\nabla^2H_q^{\ast}(g)$ admits $\lambda_N=0$ as eigenvalue with its associated normalized eigenvector $v_N:=\frac{1}{\sqrt{N}}\mathds{1}_N\in\R^N$, which means that $\mbox{rank}(\nabla^2H_q^{\ast}(g))\leq N-1$ for all $g\in\R^N$ and $q\in\Sigma_N$.
\end{proposition}
\begin{proof}
Let $g\in\R^N$, then by Theorem \ref{th:diffHq*}
\begin{align*}
\nabla^2H_q^{\ast}(g)v_N &=\frac{1}{\varepsilon}\diag(\alpha)K\frac{q}{K\alpha}-\frac{1}{\varepsilon}\diag(\alpha)K\diag\left(\frac{q}{(K\alpha)^2}\right)K\alpha\\
&=\frac{1}{\varepsilon}\diag(\alpha)K\frac{q}{K\alpha}-\frac{1}{\varepsilon}\diag(\alpha)K\frac{q}{K\alpha}=0,
\end{align*}
and $\lambda_N=0$ is an eigenvalue of $\nabla^2H_q^{\ast}(g)$.
\end{proof}

Let $(v_k)_{1\leq k\leq N}$ be the eigenvectors of $\nabla^2H_q^{\ast}(g)$, depending on both $q$ and $g$, with their respective eigenvalues $(\lambda_k)_{1\leq k\leq N}$. As the Hessian matrix is symmetric and diagonalisable, 
let us now prove that the eigenvalues associated to the eigenvectors $(v_k)_{1\leq k\leq N-1}$ of $\nabla^2 H_q^{\ast}(g)$ are all positive.
\begin{proposition}
\label{prop:eigenvalues}
For all $q\in\Sigma_N$ and $g\in\R^N$, we have that
$$0=\lambda_N<\lambda_{k} \qquad \mbox{for all} \quad 1\leq k \leq N-1.$$
\end{proposition}
\begin{proof}
The eigenvalue $\lambda_N=0$ associated to $v_N$ has been treated in Proposition \ref{prop:eigenvector_vn}. Let $v\in V=(\mbox{Vect}(v_N))^{\perp}$ (i.e. $v$ does not have constant coordinates) an eigenvector of $\nabla^2H_q^{\ast}(g)$. Hence we can suppose that, let say  $v^{(j)}$, is its larger coordinate, and that their exists $i\neq j$ such that $v^{(j)}>v^{(i)}$. Without loss of generality, we can assume that $v^{(j)}>0$. Then
\begin{align*}
[\nabla^2H_q^{\ast}(g)v]_j &=\left[\frac{1}{\varepsilon}\left(\diag\left(\diag(\alpha)K\frac{q}{K\alpha}\right)\right)v\right]_j-\left[\frac{1}{\varepsilon}\diag(\alpha)K\diag\left(\frac{q}{(K\alpha)^2}\right)K\diag(\alpha)v\right]_j\\
&=\frac{1}{\varepsilon}\alpha_jv^{(j)}\sum_{i=1}^NK_{ji}\frac{q_i}{[K\alpha]_i}-\frac{1}{\varepsilon}\sum_{i=1}^N\sum_{m=1}^N\alpha_jK_{jm}\frac{q_m}{[K\alpha]_m^2}\alpha_iK_{mi}v^{(i)}\\
&> \frac{1}{\varepsilon}\alpha_jv^{(j)}\sum_{i=1}^NK_{ji}\frac{q_i}{[K\alpha]_i}-\frac{1}{\varepsilon}\sum_{i=1}^N\sum_{m=1}^N\alpha_jK_{jm}\frac{q_m}{[K\alpha]_m^2}\alpha_iK_{mi}v^{(j)} \ \mbox{since} \ v^{(j)}\geq v^{(i)}, \forall i \\
&=0 \qquad \mbox{since} \quad \sum_{i=1}^N\alpha_iK_{im}=[K\alpha]_m.
\end{align*}
Thus $\lambda v^{(j)}=[\nabla^2H_q^{\ast}(g)v]_j>0$, and  we necessarily  have that $\lambda>0$.
\end{proof}
The set of eigenvalues of $\nabla^2 H_q^{\ast}(g)$ is also bounded from above.

\begin{proposition}
\label{prop:trace}
For all $q\in\Sigma_N$ and $g\in\R^N$ we have that
$\mbox{Tr}(\nabla^2H_q^{\ast}(g))\leq\frac{1}{\varepsilon}$
and thus $\lambda_k\leq 1/\varepsilon$ for all $k=1,\ldots,N$.
\end{proposition}
\begin{proof}
We directly get from Theorem \ref{th:diffHq*} that
$$\mbox{Tr}(\nabla^2H_q^{\ast}(g))\leq \frac{1}{\varepsilon}\mbox{Tr}\left(\diag\underbrace{\left(\diag(\alpha)K\frac{q}{K\alpha}\right)}_{\in\ \Sigma_N}\right)=\frac{1}{\varepsilon}.$$
\end{proof}

We can now provide the proof of  Theorem \ref{th:strong_convexity_H}.  Since $H_q$ is convex, proper and lower-semicontinuous, we know by the Fenchel-Moreau theorem that $H_q^{\ast\ast}=H_q$. Hence by  Corollary 12.A in the Rockafellar's book \cite{rockafellar1974conjugate}, we have  that
\begin{equation}
\nabla H_q=(\nabla H_q^{\ast})^{-1}, \label{eq:rock}
\end{equation}
in the sense that $\nabla H_q^{\ast} \circ \nabla H_q   (r) = r$ for any $r \in \Sigma_N$.

To continue the proof, we focus on a definition of the function $H_q$ restricted to the linear subspace $V$. Let $(v_1,\ldots,v_{N-1})$ be an orthonormal basis of $V=(\mbox{Vect}(v_N))^{\perp}$ and $P=[v_1\ \cdots \ v_{N-1}]\in\R^{N\times (N-1)}$ the matrix of the basis. Remark that $PP^T$ is the matrix of the orthogonal projection onto $V$, and that $PP^T=I_N-v_Nv_N^T.$ If we define $\tS_{N-1}:=P^T\Sigma_N\in\R^{N-1}$, then for $r\in\Sigma_N$, there exists $\tr\in\tS_{N-1}$ such that $r=P\tr+\frac{1}{\sqrt{N}}v_N.$
Hence we can introduce the functional $\tH_q:\tS_{N-1} \rightarrow \R$ defined by
$$\tH_q(\tr):=H_q\left(P\tr+\frac{1}{\sqrt{N}}v_N\right).$$
For $\tg\in\R^{N-1}$ we have that
\begin{align*}
\tH_q^{\ast}(\tg) &=\max_{\tr\in\tS_{N-1}}\ \langle \tg,\tr\rangle-\tH_q(\tr)\\
&=\max_{r\in\Sigma_N}\ \langle \tg,P^Tr-u_N\rangle - H_q(r) \quad \mbox{where} \ u_N=\frac{1}{N}\left(\sum_{i=1}^N v_1^{(i)},\ldots,\sum_{i=1}^Nv_{N-1}^{(i)}\right) \\
&= H_q^{\ast}(P\tg)-\langle \tg,u_N\rangle.
\end{align*}
Since $H_q^{\ast}$ is $C^{\infty}$ (see Theorem \ref{th:diffHq*}), we can differentiate $\tH_q^{\ast}$ with respect to $\tilde{g}$ to obtain that
\begin{align*}
\nabla \tH_q^{\ast}(\tg)& =P^T\nabla H_q^{\ast}(P\tg)-u_N\\
\nabla^2\tH_q^{\ast}(\tg)& = P^T\nabla^2H_q^{\ast}(P\tg)P.
\end{align*}
By Proposition \ref{prop:eigenvalues}, we know that $\nabla^2H_q^{\ast}(P\tg)\in\R^{N\times N}$ admits a unique eigenvalue equals to $0$ which is associated to the eigenvector $v_N$. All  other eigenvalues are  positive (Proposition \ref{prop:eigenvalues}) and bounded from above by $1/\varepsilon$ (Proposition \ref{prop:trace}). Since $\nabla\tH_q^{\ast}  : \R^{(N-1)} \to \R^{(N-1)}$ is a $C^{\infty}$-diffeomorphism, using equality \eqref{eq:rock} (that is also valid for $\tH_q$), we have   that
\begin{eqnarray*}
\nabla^2 \tilde{H}_q(\tilde{r}) & = &  \nabla \left(  (\nabla \tH_q^{\ast})^{-1}(\tilde{r}) \right)\\
& = & [\nabla^2 \tilde{H}_q^{\ast}((\nabla \tH_q^{\ast})^{-1}(\tilde{r}))]^{-1} \\
& = & [\nabla^2 \tilde{H}_q^{\ast}(\nabla \tilde{H}_q(\tilde{r}))]^{-1},
\end{eqnarray*}
where the second equality follows from the global inversion theorem, and the last one again uses  equality \eqref{eq:rock}.
Thus we get
$$\lambda_{\min}(\nabla^2\tH_q(\tr))\geq\varepsilon.$$
The above inequality implies  the strong convexity of $\tH_q$ which reads for $\tr_0,\tr_1\in\tS_{n-1}$
$$\tH_q(\tr_1)\geq \tH_q(\tr_0)+\nabla\tH_q(\tr_0)^T(\tr_1-\tr_0)+\frac{\varepsilon}{2}\|\tr_1-\tr_0\|^2,$$
and this translates for $H_q$ and $r_0,r_1\in\Sigma_N$ to
$$H_q(r_1)\geq H_q(r_0)+\nabla H_q(r_0)^TPP^T(r_1-r_0)+\frac{\varepsilon}{2}\| PP^T(r_1-r_0)\|^2.$$
 To conclude, we remark that  $(r_1-r_0)\in V$ (indeed one has that $r_1-r_0=\sum_{j=1}^{N-1}\langle v_j,r_1-r_0\rangle v_j$ since $\langle v_N,r_1-r_0\rangle=0$ and thus $PP^T(r_1-r_0)=r_1-r_0$). Hence, we finally obtain the strong convexity of $H_q$
$$H_q(r_1)\geq H_q(r_0)+\nabla H_q(r_0)^T(r_1-r_0)+\frac{\varepsilon}{2}\| (r_1-r_0)\|^2.$$
This completes the proof of Theorem \ref{th:strong_convexity_H}.

\section{Lipschitz constant of \texorpdfstring{$H_q$}{H_q} -  Proof of Lemma \ref{lemma:Lipschitz}}
\label{sec:dual_bounded}
The dual version of the minimization problem \eqref{def:primal_Sinkhorn} is given in \cite{cuturi2013fast} by
\begin{equation}
\label{eq:dual_Sinkhorn}
W_{2,\varepsilon}^2(r,q)=\underset{\alpha,\beta\in\R^N}{\max} \ \alpha^Tr+\beta^Tq-\sum_{1\leq m,\ell \leq N}\varepsilon e^{-\frac{1}{\varepsilon}(C_{m\ell}-\alpha_m-\beta_\ell)}
\end{equation}
where $C_{m\ell}$ are the entries of the matrix cost $C$. We recall the notation
$$
\Sigma_N^{\rho} = \left\{ r \in \Sigma_N \; : \; \min_{1 \leq \ell \leq N} r_{\ell} \geq \rho \right\}  \mbox{ for some } 0 < \rho < 1. 
$$
We now recall the Lemma \ref{lemma:Lipschitz}.
\begin{lemma}\label{lemma:Lipschitz0}
Let $q \in \Sigma_N$ and $0 < \rho < 1$. Then, one has that $r \mapsto H_q(r)$ is $L_{\rho,\varepsilon}$-Lipschitz on $\Sigma_N^{\rho}$ with
\begin{equation}
L_{\rho,\varepsilon}= \left( \sum_{1 \leq m \leq N} \left(2 \varepsilon \log(N) +\adjustlimits\inf_{1 \leq k \leq N} \sup_{1 \leq \ell\leq N} |C_{m\ell} - C_{k \ell}|    - \varepsilon \log(\rho) \right)^2  \right)^{1/2}. \label{eq:defL2}
\end{equation}
\end{lemma}
\begin{proof}
Let $r,s,q\in\Sigma_N$. We denote by $(\alpha^{q,r},\beta^{q,r})$ a pair of optimal dual variables in the problem \eqref{eq:dual_Sinkhorn}. Then, we have that
\begin{align}
\vert H_q(r)-H_q(s)\vert &=(H_q(r)-H_q(s))\mathds{1}_{H_q(r)\geq H_q(s)}+(H_q(s)-H_q(r))\mathds{1}_{H_q(r)\leq H_q(s)} \nonumber \\
& \leq \left(\langle\alpha^{q,r},r\rangle+\langle\beta^{q,r},q\rangle -\sum_{m,\ell}\varepsilon e^{-\frac{1}{\varepsilon}(C_{m\ell}-\alpha_m^{q,r}-\beta_\ell^{q,r})}-\langle\alpha^{q,r},s\rangle-\langle\beta^{q,r},q\rangle +\right. \nonumber \\
& \qquad \left. \sum_{m,\ell}\varepsilon e^{-\frac{1}{\varepsilon}(C_{m\ell}-\alpha_m^{q,r}-\beta_\ell^{q,r})}\right)\mathds{1}_{(H_q(r)\geq H_q(s))} \nonumber\\
&+\left(\langle\alpha^{q,s},s\rangle+\langle\beta^{q,s},q\rangle -\sum_{m,\ell}\varepsilon e^{-\frac{1}{\varepsilon}(C_{m\ell}-\alpha_m^{q,s}-\beta_\ell^{q,s})}-\langle\alpha^{q,s},r\rangle-\langle\beta^{q,s},q\rangle +\right. \nonumber\\
& \qquad \left. \sum_{m,\ell}\varepsilon e^{-\frac{1}{\varepsilon}(C_{m\ell}-\alpha_m^{q,s}-\beta_\ell^{q,s})}\right)\mathds{1}_{(H_q(r)\leq H_q(s))} \nonumber\\
& \leq \underset{\alpha\in\{\alpha^{q,r},\alpha^{q,s}\}}{\sup}\vert \langle \alpha,r-s\rangle\vert  \leq \underset{\alpha \in\{\alpha^{q,r},\alpha^{q,s}\}}{\sup} \vert \alpha\vert \ \vert r-s\vert. \label{eq:Lip}
\end{align}
Let us now prove that the norm of the dual variable $\alpha^{q,r}$ (resp.\  $\alpha^{q,s}$)  is bounded by a constant not depending on $q$ and $r$ (resp.\ $q$ and $s$). To this end, we follow some of the arguments in the proof of Proposition A.1 in \cite{2016-genevay-nips}. Since the dual variable $\alpha^{q,r}$ achieves the maximum in equation \eqref{eq:dual_Sinkhorn}, we have that for any $1 \leq m \leq N$
$$
r_m-\sum_{1\leq \ell \leq N} e^{-\frac{1}{\varepsilon}(C_{m\ell}-\alpha^{q,r}_m-\beta^{q,r}_\ell)} = 0.
$$
Let $r \in \Sigma_N^{\rho}$. Hence, $r_m \neq 0$, and thus one may define $\lambda_{m} = \varepsilon \log(r_{m})$. Then, it follows from the above equality that $\sum_{1\leq \ell \leq N} e^{-\frac{1}{\varepsilon}(C_{m\ell} + \lambda_{m}-\alpha^{q,r}_m-\beta^{q,r}_\ell)} = 1$ which implies that
$$
\alpha^{q,r}_m = -\varepsilon \log \left( \sum_{1\leq \ell \leq N} e^{-\frac{1}{\varepsilon}(C_{m\ell} + \lambda_{m}-\beta^{q,r}_{\ell})} \right) .
$$
Now, for each $1 \leq m \leq N$, we define
\begin{equation}
\tilde{\alpha}^{q,r}_m = \min_{1 \leq \ell \leq N}\left\{ C_{m\ell} + \lambda_{m}-\beta^{q,r}_{\ell} \right\} =  \min_{1 \leq \ell \leq N}\left\{ C_{m\ell} -\beta^{q,r}_{\ell} \right\} + \lambda_{m}, \label{eq:ctransform}
\end{equation}
and we consider the inequality
\begin{equation}
|\alpha^{q,r}_m- \alpha^{q,r}_k| \leq |\alpha^{q,r}_m - \tilde{\alpha}^{q,r}_m| + |\tilde{\alpha}^{q,r}_m - \tilde{\alpha}^{q,r}_k| + |\tilde{\alpha}^{q,r}_k - \alpha^{q,r}_k |. \label{eq:decompineq}
\end{equation}
By equation \eqref{eq:ctransform} one has that $\tilde{\alpha}^{q,r}_m + \beta^{q,r}_{\ell} - C_{m\ell} - \lambda_{m} \leq 0$. Hence we get
\begin{equation}
- \alpha^{q,r}_m = \varepsilon \log \left( \sum_{1\leq \ell \leq N} e^{-\frac{1}{\varepsilon} \tilde{\alpha}^{q,r}_m }e^{\frac{1}{\varepsilon}(\tilde{\alpha}^{q,r}_m + \beta^{q,r}_{\ell} - C_{m\ell} - \lambda_{m})} \right) \\
\leq - \tilde{\alpha}^{q,r}_m +  \varepsilon \log(N). \label{eq:softmax1}
\end{equation}
On the other hand, using the inequality
$$
\sum_{1\leq \ell \leq N} e^{-\frac{1}{\varepsilon}(C_{m\ell} + \lambda_{m}-\beta^{q,r}_{\ell})} \geq e^{-\frac{1}{\varepsilon}(C_{m \ell_{\ast}} + \lambda_{m}-\beta^{q,r}_{\ell_{\ast}})} = e^{-\frac{1}{\varepsilon} \tilde{\alpha}^{q,r}_m},
$$
where $\ell_{\ast}$ is a value of $1 \leq \ell \leq N$ achieving the minimum in \eqref{eq:ctransform}, we obtain that
\begin{equation}
- \alpha^{q,r}_m \geq - \tilde{\alpha}^{q,r}_m.  \label{eq:softmax2}
\end{equation}
By combining inequalities \eqref{eq:softmax1} and \eqref{eq:softmax2}, we finally have
\begin{equation}
|\tilde{\alpha}^{q,r}_m - \alpha^{q,r}_m| \leq \varepsilon \log(N). \label{eq:softmax3}
\end{equation}
To conclude, it remains to remark that, by equation \eqref{eq:ctransform}, the vector $(\tilde{\alpha}^{q,r}_m - \lambda_{m})_{1 \leq m \leq N}$ is the c-transform of the vector $(\beta^{q,r}_{\ell})_{1 \leq \ell \leq N}$ for the cost matrix $C$. Therefore, by using  standard results  in optimal transport which relate c-transforms to the modulus of continuity of the cost (see e.g.\ \cite{Santam15}, p.\ 11) one obtains that
$$
|\tilde{\alpha}^{q,r}_m - \tilde{\alpha}^{q,r}_k + \lambda_{k} - \lambda_{m}| \leq \sup_{1 \leq \ell \leq N} |C_{m\ell} - C_{k\ell}|,
$$
which implies that
\begin{equation}
|\tilde{\alpha}^{q,r}_m - \tilde{\alpha}^{q,r}_k| \leq \sup_{1 \leq \ell \leq N} |C_{m\ell} - C_{k \ell}| +  \varepsilon | \log(r_{m}) - \log(r_{k}) |. \label{eq:modulus}
\end{equation}
By combining the upper bounds \eqref{eq:softmax3} and \eqref{eq:modulus} with the decomposition \eqref{eq:decompineq} we finally come to the inequality 
$$
|\alpha^{q,r}_m - \alpha^{q,r}_k| \leq 2 \varepsilon \log(N) + \sup_{1 \leq \ell \leq N} |C_{m\ell} - C_{k \ell}|   + \varepsilon | \log(r_{m}) - \log(r_{k}) |.
$$
Since this inequality remains true for each $k$, and that the dual variables achieving the maximum in equation \eqref{eq:dual_Sinkhorn} are  defined up to an additive constant, one may assume that $\alpha^{q,r}_k = 0$ choosing the $k$ minimizing $\sup_{1 \leq \ell \leq N} |C_{m\ell} - C_{k \ell}|$. Under such a condition, we finally obtain that
\begin{equation*}
|\alpha|  \leq  \left( \sum_{1 \leq m \leq N} \left(2 \varepsilon \log(N) +  \inf_{1 \leq k \leq N} \left\{ \sup_{1 \leq \ell \leq N} |C_{m\ell} - C_{k \ell}|  + \varepsilon | \log(r_{m}) - \log(r_{k}) |  \right\}\right)^2   \right)^{1/2} .
\end{equation*}

Using inequality \eqref{eq:Lip} and the assumption that $r \in \Sigma_N^{\rho}$ (in particular $\rho\leq r_m/r_k\leq 1/\rho$), we can thus conclude that $r \mapsto H_q(r)$ is $L_{\rho,\varepsilon}$-Lipschitz on $\Sigma_N^{\rho}$ for
\begin{equation}
L_{\rho,\varepsilon}= \left( \sum_{1 \leq m \leq N} \left(2 \varepsilon \log(N) + \adjustlimits\inf_{1 \leq k \leq N} \sup_{1 \leq \ell\leq N} |C_{m\ell} - C_{k \ell}|    -  \varepsilon \log(\rho) \right)^2  \right)^{1/2}.
\end{equation}
\end{proof}

\section{Useful concentration inequalities} \label{app:concen}

We state and prove below the various concentration inequalities that have been used in the proof of Theorem \ref{theo:oracle}.

\begin{proposition} \label{prop:concentration1} 
For any $u > 0$
\begin{equation}
\P \left( \vert \hat{\br}_{n}^{\varepsilon}-r^{\varepsilon}\vert > (1+u)  \frac{2 \sqrt{2} L_{\rho,\varepsilon}}{\varepsilon \sqrt{n}}  \right) \leq \exp\left(  -u^2\right), \label{eq:ineqconcen1}
\end{equation}
where $\br_n^{\varepsilon}$ is the Sinkhorn barycenter of the iid random measures $\bq_1,\ldots,\bq_n$ (assumed to belong to  $\Sigma_N^{\rho}$) as defined by \eqref{eq:defrneps}.
\end{proposition}

\begin{proof}
To derive a concentration inequality for the random variable
$$
Z = \vert \hat{\br}_{n}^{\varepsilon}- r^{\varepsilon} \vert,
$$
we write it as  $Z = f(\bq_1,\ldots,\bq_n)$ where $f : \Sigma_N^{\rho} \times \ldots \times \Sigma_N^{\rho}  \to \R$ is a measurable function, and we shall prove that $f$ satisfies the following bounded difference inequality
\begin{equation}
| f(\bq_1,\ldots,\bq_{i},\ldots,\bq_n) -  f(\bq_1,\ldots,\bq_{i}',\ldots,\bq_n) | \leq  \frac{4 L^2_{\rho,\varepsilon}}{\varepsilon n}, \quad   \mbox{ for any } 1 \leq i \leq n, \label{eq:bounddiff}
\end{equation}
where $\bq_{i}'$ denotes an independent random measure of the sample $\bq_1,\ldots,\bq_n$ that is distributed as $\bq$. To this end, we introduce the following empirical versions of  the function
\begin{equation}
r \mapsto F(r) := \E_{\bq\sim\P}[W_{2,\varepsilon}^2(r,\bq)] \label{eq:defF}
\end{equation}
that are defined as
$$
\hat{F}(r)  = \frac{1}{n}\sum_{i=1}^n W_{2,\varepsilon}^2(r,{\bq}_i)  \quad \mbox{and} \quad \hat{F}^{(i)}(r)  = \frac{1}{n} \left( \sum_{j=1, j \neq i }^n W_{2,\varepsilon}^2(r,{\bq}_j) +  W_{2,\varepsilon}^2(r,{\bq}_{i}')\right),
$$
and we introduce the quantities
$$
 \br_n^{\varepsilon,(i)} =\uargmin{r\in\Sigma_N^{\rho}} \hat{F}^{(i)}(r) \quad \mbox{and} \quad Z^{(i)} = \vert \br_n^{\varepsilon,(i)} - r^{\varepsilon} \vert.
$$
Then, we proceed as in the proof  of Theorem 6 in \cite{shalev2009stochastic}. First, we remark that
\begin{eqnarray}
\hat{F}(\br_n^{\varepsilon,(i)}) - \hat{F}(\br_n^{\varepsilon}) & = & \frac{W_{2,\varepsilon}^2(\br_n^{\varepsilon,(i)},{\bq}_i) - W_{2,\varepsilon}^2(\br_n^{\varepsilon},{\bq}_i) }{n} + \frac{1}{n} \sum_{j=1, j \neq i }^n W_{2,\varepsilon}^2(\br_n^{\varepsilon,(i)},{\bq}_j) - W_{2,\varepsilon}^2(\br_n^{\varepsilon},{\bq}_j) \nonumber \\
& = & \frac{W_{2,\varepsilon}^2(\br_n^{\varepsilon,(i)},{\bq}_i) - W_{2,\varepsilon}^2(\br_n^{\varepsilon},{\bq}_i) }{n}  + \frac{W_{2,\varepsilon}^2(\br_n^{\varepsilon},{\bq}_{i}') - W_{2,\varepsilon}^2(\br_n^{\varepsilon,(i)},{\bq}_{i}') }{n} \nonumber  \\
& & + \hat{F}^{(i)}(\br_n^{\varepsilon,(i)}) - \hat{F}^{(i)}(\br_n^{\varepsilon})\nonumber   \\
& \leq & \frac{|W_{2,\varepsilon}^2(\br_n^{\varepsilon,(i)},{\bq}_i) - W_{2,\varepsilon}^2(\br_n^{\varepsilon},{\bq}_i) |}{n}  + \frac{|W_{2,\varepsilon}^2(\br_n^{\varepsilon},{\bq}_{i}') - W_{2,\varepsilon}^2(\br_n^{\varepsilon,(i)},{\bq}_{i}') |}{n} \nonumber  \\
& \leq & \frac{2 L_{\rho,\varepsilon}}{n}  \vert \br_n^{\varepsilon,(i)} - \br_n^{\varepsilon} \vert, \label{eq:bound1}
\end{eqnarray}
where the first inequality follows from the fact that $ \br_n^{\varepsilon,(i)} $ is a minimizer of $\hat{F}^{(i)}$, and the second one from Lemma \ref{lemma:Lipschitz} on the Lipschitz continuity of $r \mapsto W_{2,\varepsilon}^2(r,q)$. Now, using Theorem \ref{th:strong_convexity_H}, one has that the function
$\hat{F}$ is $\varepsilon$-strongly convex, which implies that
\begin{equation}
\hat{F}(\br_n^{\varepsilon,(i)}) - \hat{F}(\br_n^{\varepsilon}) \geq \frac{\varepsilon}{2}  \vert \br_n^{\varepsilon,(i)} - \br_n^{\varepsilon} \vert^2. \label{eq:bound2}
\end{equation}
Combining \eqref{eq:bound1} and \eqref{eq:bound2}, we obtain that $  \vert \br_n^{\varepsilon,(i)} - \br_n^{\varepsilon} \vert \leq  \frac{4 L_{\rho,\varepsilon}}{\varepsilon n} $, and note that by the Lipschitz continuity of $r \mapsto W_{2,\varepsilon}^2(r,q)$, it follows that, for any $q \in \Sigma_N^{\rho}$,
\begin{equation}
|W_{2,\varepsilon}^2(\br_n^{\varepsilon,(i)},q) - W_{2,\varepsilon}^2(\br_n^{\varepsilon},q) | \leq \frac{4 L^2_{\rho,\varepsilon}}{\varepsilon n}.  \label{eq:bound3}
\end{equation}
By the triangle inequality we finally obtain that
$$
|Z - Z^{(i)}| = | \vert \br_n^{\varepsilon} - r^{\varepsilon} \vert -  \vert \br_n^{\varepsilon,(i)} - r^{\varepsilon} \vert | \leq \vert \br_n^{\varepsilon,(i)} - \br_n^{\varepsilon} \vert \leq  \frac{4 L_{\rho,\varepsilon}}{\varepsilon n},
$$
which proves that inequality \eqref{eq:bounddiff} holds.
Now, using concentration of measure for a function of random variables satisfying the bounded difference inequality \eqref{eq:bounddiff}, we obtain that, for any $t > 0$ (see e.g. Theorem 6.2 in \cite{MR3185193})
\begin{equation}
\P \left( \vert \hat{\br}_{n}^{\varepsilon}- r^{\varepsilon} \vert - \E [ \vert \hat{\br}_{n}^{\varepsilon}- r^{\varepsilon} \vert ] > t \right) \leq \exp\left(  -\frac{n\varepsilon^2 t^2}{8 L^2_{\rho,\varepsilon}} \right). \label{eq:concen1}
\end{equation}
To conclude the proof it remains to obtain an upper bound on $ \E [ \vert \hat{\br}_{n}^{\varepsilon}- r^{\varepsilon} \vert ]$. We start with the basic inequality
\begin{equation}
 \E [ \vert \hat{\br}_{n}^{\varepsilon}- r^{\varepsilon} \vert ] \leq \sqrt{ \E [ \vert \hat{\br}_{n}^{\varepsilon}- r^{\varepsilon} \vert^2 ]}. \label{eq:basic}
\end{equation}
By Theorem \ref{th:strong_convexity_H}, it follows that the function $F$ defined in \eqref{eq:defF} is $\varepsilon$-strongly convex for the Euclidean 2-norm. Hence, as $r^{\varepsilon}$ is by definition a minimizer of $F$, one has that
\begin{equation}
\vert \hat{\br}_{n}^{\varepsilon}-r^{\varepsilon}\vert^2 \leq \frac{2}{\varepsilon} \left( F(\hat{\br}_{n}^{\varepsilon}) - F(r^{\varepsilon}) \right), \label{eq:Lip1}
\end{equation}
Then, for any $1 \leq i \leq n$, since ${\bq}_i'$ is an independent copy of ${\bq}_i$,  we clearly have that
$$
 \E [ F(\hat{\br}_{n}^{\varepsilon}) ] =  \E [ F(\br_n^{\varepsilon,(i)}) ] =  \E [  W_{2,\varepsilon}^2(\br_n^{\varepsilon,(i)},{\bq}_i)  ] 
$$
Hence, we may write
$
 \E [ F(\hat{\br}_{n}^{\varepsilon}) ] = \frac{1}{n} \sum_{i=1}^{n}  \E [  W_{2,\varepsilon}^2(\br_n^{\varepsilon,(i)},{\bq}_i)  ],
$
and since one also has that
$
 \E [ \hat{F}(\hat{\br}_{n}^{\varepsilon}) ] = \frac{1}{n} \sum_{i=1}^{n}  \E [  W_{2,\varepsilon}^2(\hat{\br}_{n}^{\varepsilon},{\bq}_i)  ],
$
inequality \eqref{eq:bound3} yields
$$
\E [F(\hat{\br}_{n}^{\varepsilon}) -  \hat{F}(\hat{\br}_{n}^{\varepsilon})] \leq \frac{4 L^2_{\rho,\varepsilon}}{\varepsilon n}.
$$
Finally, as $F(r^{\varepsilon}) = \E [ \hat{F}(r^{\varepsilon}) ] \geq \E [ \hat{F}(\hat{\br}_{n}^{\varepsilon}) ]$ we obtain from the above inequality that
\begin{equation}
 \E [ F(\hat{\br}_{n}^{\varepsilon}) ] - F(r^{\varepsilon}) \leq \frac{4 L^2_{\rho,\varepsilon}}{\varepsilon n}. \label{eq:concen2}
\end{equation}
Combining inequalities  \eqref{eq:concen1},  \eqref{eq:basic}  \eqref{eq:Lip1} and  \eqref{eq:concen2} allows to conclude the proof of Proposition \ref{prop:concentration1}. 

\end{proof}

\begin{proposition} \label{prop:concentration2} 
For any $u > 0$
\begin{equation}
\P \left(  \vert {\br}_n^{\varepsilon}-\hat{{\br}}^{\varepsilon}_{n,p}\vert^2 > \frac{2L_{\rho,\varepsilon}}{\varepsilon}\left(u + \rho( N + \sqrt{N}) \right)  \right) \leq 2^N \sum_{i=1}^{n} \exp\left(  - p_i u^2 \right), \label{eq:ineqconcen2}
\end{equation}
where $\br_n^{\varepsilon}$ is the Sinkhorn barycenter of the iid random measures $\bq_1,\ldots,\bq_n$ (assumed to belong to  $\Sigma_N^{\rho}$) as defined by \eqref{eq:defrneps}. 
\end{proposition}

\begin{proof}
By arguing as in the proof of Theorem \ref{th:Sinkhorn_rate}, we have the following inequality
\begin{equation}
 \vert {\br}_n^{\varepsilon}-\hat{{\br}}^{\varepsilon}_{n,p}\vert^2  \leq \frac{2L_{\rho,\varepsilon}}{\varepsilon} \left( \frac{1}{n}\sum_{i=1}^n \vert \bq_i-\tilde{\bq}_i^{p_i}\vert +  \rho( N + \sqrt{N}) \right). \label{eq:mulinom1}
\end{equation}
Then, conditionally on $\bq_i$, we recall that $p_i \tilde{\bq}_i^{p_i}$ is a random vector following a multinomial distribution $\mathcal{M}(p_i,\bq_i)$. Hence, using the fact that the Euclidean 2-norm satisfies
$
\vert \bq_i-\tilde{\bq}_i^{p_i}\vert \leq \sum_{k=1}^{N} |\bq_{i,k} - \tilde{\bq}_{i,k}^{p_i}|,
$
it follows from the so-called Bretagnolle-Huber-Carol inequality (see Proposition A6.6 in \cite{MR1385671}) and by conditioning on $\bq_i$ that, for any $u > 0$,
\begin{equation}
\P \left( \vert \bq_i-\tilde{\bq}_i^{p_i}\vert  \geq u  \right) \leq 2^N \exp \left( - \frac{p_i u^2}{2}\right).
\label{eq:mulinom2}
\end{equation}
Hence, combining inequalities \eqref{eq:mulinom1} and \eqref{eq:mulinom2} with the union bound $\P \left(  \frac{1}{n}\sum_{i=1}^n \vert \bq_i-\tilde{\bq}_i^{p_i}\vert \geq u\right) \leq \sum_{i=1}^n \P \left(  \vert \bq_i-\tilde{\bq}_i^{p_i}\vert \geq u\right)  $ allows to complete the proof of Proposition \ref{prop:concentration2}.

\end{proof}

\section{Algorithms to compute penalized Wasserstein barycenters of Section \ref{sec:Regbar}}\label{sec:algo}
In this section we describe how the minimization problem 
\begin{equation}
\label{eq:regbar}
\min_{\mu}  \frac{1}{n} \sum_{i=1}^{n} W_2^2(\mu,\nu_{i}) +\gamma E(\mu) \mbox{ over } \mu \in \PP_{2}(\Omega),
\end{equation}
can be solved numerically by using an appropriate discretization to compute a numerical approximation of a regularized Wasserstein barycenter and the work of \cite{CuturiPeyre}.
In our numerical experiments, we focus on the case  where $E(\mu) = + \infty$ if $\mu$ is not a.c.\ to enforce the regularized Wasserstein barycenter to have a smooth pdf (we write $E(f) = E(\mu_{f})$ if $\mu$ has a density $f$).  In this setting, if the grid of points is of sufficiently large size, then the weights $f^k$ yield a good approximation of this pdf.
A discretization of the minimization problem \eqref{eq:regbar} is used to compute a numerical approximation of a regularized Wasserstein barycenter $\bmu_{\P_n}^{\gamma}$. It consists of using a fixed grid $\{x^k\}_{k=1}^N$  of equally spaced points $x^k\in\mathbb{R}^d$, and to approximate $\bmu_{\P_n}^{\gamma}$ by the discrete measure $\sum_{k=1}^{N} f^k \delta_{x^k}$ where the $f^k$ are positive weights summing to one which minimize a discrete version of the optimization problem  \eqref{eq:regbar}.

In what follows, we first describe an algorithm that is specific to the one-dimensional case, and then we propose another algorithm that is valid for any $d \geq 1$.

\paragraph{Discrete algorithm for $d=1$ and data defined on the same grid} \label{sec:algo1D}

We first propose to compute a regularized empirical Wasserstein barycenter for a dataset made of discrete measures $\nu_1,\ldots,\nu_{n}$ (or one-dimensional histograms) defined on the same grid of reals $\{x^k\}_{k=1}^N$ that the one chosen to approximate $\bmu_{\P_n}^{\gamma}$. Since the grid is fixed, we identify a discrete measure $\nu$ with the vector of weights $\nu=(\nu(x^1),\ldots,\nu(x^N))$ in $\R^N_{+}$ (with entries that sum up to one) of its values on this grid.

 The estimation of the regularized barycenter onto this grid can be formulated as:
\begin{equation}
\label{eq:regdens_app}
\min_f  \frac{1}{n} \sum_{i=1}^{n} W_2^2({f},\nu_{i}) +\gamma E({f}) \mbox{ s.t } \sum_k f^k=1,\mbox{ and } f^k = f(x^{k})\geq 0,
\end{equation}
with the obvious abuse of notation $W_2^2({f},\nu_{i})  = W_2^2(\mu_{f},\nu_{i})$ and $E(f) = E(\mu_{f})$.

Then, to compute a minimizer of the convex optimization problem \eqref{eq:regdens_app}, we perform a subgradient descent. We denote by $(f^{(\ell)})_{\ell \geq 1}$ the resulting sequence of discretized regularized barycenters  in $\R^N$ along the descent. Hence, given an initial value $f^{(1)} \in \R^N_{+}$ and for $\ell\ge 1$, we thus have
\begin{equation}
f^{{(\ell+1)}}=\Pi_S \left(f^{(\ell)}-\tau^{(\ell)}\left[\gamma\nabla E(f^{(\ell)})+\frac{1}{n}\sum_{i=1}^n\nabla_1 W_2^2({f^{(\ell)}},{\nu_i})\right]\right) \label{eq:algo}
\end{equation}
where $\tau^{(\ell)}$ is the $\ell$-th step time, and $\Pi_S$ stands for the projection on the simplex
$S=\{y \in\R^N_{+} \ \text{such that} \ \sum_{j = 1}^{N} y^j =1\}.$
Thanks to Proposition 5 in \cite{peyre2012wasserstein}, we are able to compute a sub-gradient of the squared Wasserstein distance $W_2^2({f^{(\ell)}},{\nu_i})$ with respect to its first argument (for discrete distributions). For that purpose, we denote by $R_f(s)=\sum_{x^j\le s}f(x^j)$ the cdf of $\mu_{f} = \sum_{k=1}^{N} f(x^k) \delta_{x^k}$ and by $R_f^-(t)=\inf\{s\in\R: R_f(s)\ge t\}$ its pseudo-inverse.
\begin{proposition}[\cite{peyre2012wasserstein}]
Let $f=(f(x^1),f(x^2),\, \ldots,\, f(x^N))$ and $\nu=(\nu(x^1), \nu(x^2),\,\ldots,\,\nu(x^N))$ be two discrete distributions defined on the same grid of values $x^1,\ldots,x^N$ in $\R$. For $p\ge 1$, the subgradients of $f\mapsto W_p^p(f,\nu)$ can be written as
\begin{equation}
\nabla_1W_p^p(f,\nu):x_j\mapsto\sum_{m\ge j}\vert x^m-\tilde{x}^m\vert^p-\vert x^{m+1}-\tilde{x}^m\vert^p
\end{equation}
where
\[\left\{
\begin{array}{ll}
& \tilde{x}^m=x^k \ \text{if} \ R_g(x^{k-1})<R_f(x^m)<R_\nu(x^k)\\
& \tilde{x}^m\in [x^{k-1},x^k] \ \text{if} \ R_f(x^m)=R_\nu(x^k)
\end{array}
\right.\]
\end{proposition}
Even if subgradient descent is only shown to converge with diminishing time steps \cite{Boyd08}, we observed  that using a small fixed step time (of order $10^{-5}$) is sufficient to obtain in practice a convergence of the iterates $(f^{(\ell)})_{\ell \geq 1}$. Moreover, we have noticed that the principles of FISTA (Fast Iterative Soft Thresholding, see e.g.\ \cite{Beck2009FISTA})  accelerate the speed of convergence of the above described algorithm.

\paragraph{Discrete algorithm for $d \geq 1$ in the general case} \label{sec:algo2D}

We assume that data $\nu_1,\ldots,\nu_n$ are given in the form of $n$ discrete probability measures (histograms) supported on $\R^{d}$ (with $d \geq 1$) that are not necessarily defined on the same grid. More precisely, we assume that
$$
\nu_i =  \sum_{j=1}^{p_{i}} \nu_i^j \delta_{y_{i}^{j}}
$$
for $1 \leq i \leq n$ where the $y_{i}^{j}$'s are arbitrary locations in $\Omega \subset \R^{d}$, and the $\nu_i^j$'s are positive weights (summing up to one for each $i$).

The estimation of the regularized barycenter onto a given grid $\{x^k\}_{k=1}^N$  of  $\mathbb{R}^d$ can then be formulated as the following minimization problem:
\begin{equation}
\label{eq:regdens_app2}
\min_f  \frac{1}{n} \sum_{i=1}^{n} W_2^2({f},\nu_{i}) +\gamma E({f}) \mbox{ s.t } \sum_k f^k=1,\mbox{ and } f^k\geq 0,
\end{equation}
with the notation $f=(f^1,f^2,\, \ldots,\, f^N)$ and the convention that $W_2^2({f},\nu_{i})$ denotes the squared Wasserstein distance between $\mu_{f} = \sum_{k=1}^{N} f^k \delta_{x^k}$ and $\nu_i$.

Problem \eqref{eq:regdens_app2} could be exactly solved by considering the discrete $p_{i} \times N$ transport matrices $S_i$ between  the barycenter $\mu_{f}$ to estimate and the data $\nu_i$. Indeed, problem \eqref{eq:regdens_app2} is equivalent to the convex problem 
\begin{equation}
\label{eq:regdens_app0}
\min_f \min_{S_1 \cdots S_n} \frac{1}{n} \sum_{i=1}^{n} \sum_{j=1}^{p_{i}}\sum_{k=1}^N ||y_{i}^{j}-x^k||^2 S_i^{j,k} +\gamma E({f}) 
\end{equation}
 under the linear constraints
$$\forall i=1,\ldots,n,\,  \sum_{j=1}^{p_{i}} S_i^{j,k}= f^k,\, \sum_{k=1}^N S_i^{j,k}= \nu^j_i,\mbox{ and } S_i^{j,k}\geq 0.$$
   However, optimizing over the $p_{i} \times N$ transport matrices $S_i$ for $1 \leq i \leq n$ involves memory issues when using an accurate discretization grid $\{x^k\}_{k=1}^N$ with a  large value of $N$. For this reason, we consider subgradient descent algorithms that allow dealing directly with problem \eqref{eq:regdens_app2}.

To this end, we rely on the dual approach introduced in \cite{carlier2014numerical} and the numerical optimisation scheme proposed in \cite{CuturiPeyre}.
Following these works, one can show that the dual problem  of \eqref{eq:regdens_app2} with a regularization of the form $E(Kf)$ and $K$ a discrete linear operator reads as
\begin{equation}
\label{eq:regdens2}
\min_{\phi_0,\cdots\phi_{n}} \sum_{i=1}^n H_{\nu_i}(\phi_i)+E^*_\gamma(\phi_{0})\,\,  \mbox{ s.t } K^T\phi_0+\sum_{i=1}^{n}\phi_i=0,
\end{equation}
where the $\phi_i$'s are dual variables (vectors in $\R^N$) defined on the discrete grid $\{x^k\}_{k=1}^N$, $E^*_\gamma$ is the Legendre transform of $\gamma E$ and $H_{\nu_i}(.)$ is the Legendre transform of $W_2^2(.,\nu_{i})$ that reads:
$$H_{\nu_i}(\phi_i)=\sum_{j=1}^{p_{i}} \nu_i^j  \min_{k=1\cdots N} \left(\frac12 ||y^{j}_{i}-x^k||^2-\phi_i^k\right).$$
Barycenter estimations $f_i$ can finally be recovered from the optimal dual variables $\phi_i$ solution of \eqref{eq:regdens2}  as:
\begin{equation}\label{eq:bar_reconstr}f_i\in\partial H_{\nu_i}(\phi_i),\,\,\textrm{for } i=1\cdots n.\end{equation}
Following \cite{carlier2014numerical}, one value of the above subgradient can be  obtained at point $x^k$ as:
\begin{equation}
\label{eq:subgrad}\partial H_{\nu_i}(\phi_i)_k=\sum_{j=1}^{p_{i}}  \nu^j_iS_i^{j,k},\end{equation}
where $S_i^{j,k}$ is any row stochastic matrix of size $p_{i} \times N$ checking:
$$S_i^{j,k}\neq 0\textrm{ iff }k\in\uargmin{k=1\cdots N} \left(\frac12 ||y^{j}_{i}-x^k||^2-\phi_i^k\right).$$
From the previous expressions, we see that $f_i^k=\sum_{j=1}^{p_{i}} \nu^j_iS_i^{j,k}$ corresponds to the discrete pushforward of data $\nu_i$ with the transport matrix $S_i$ with the associated cost:
$$ H_{\nu_i}(\phi_i)= \sum_{j=1}^{p_{i}}  \sum_{k=1}^{N} \left(\frac12 ||y^{j}_{i}-x^k||^2-\phi_i^k\right) S_i^{j,k}\nu^j_i.$$
\paragraph{Numerical optimization}
Following \cite{CuturiPeyre}, the dual problem \eqref{eq:regdens2}, can be simplified by removing one variable and thus  discarding the linear constraint  $ K^T\phi_0+\sum_{i=1}^{n}\phi_i=0$. In order to inject the regularity given by $\phi_0$ in all the reconstructed barycenters obtained by $\phi_i$, $i=1\cdots n$, we modified the change of variables of \cite{CuturiPeyre} by setting $\psi_{i}=\phi_i+K^T\phi_{0}/n$ for $i=1\cdots n$ and $\psi_{0}=\phi_0$, leading to $\sum_{i=1}^{n}\psi_{i}=0$. One  variable, say $\psi_n$, can then be directly obtained from the other ones. Observing that $\phi_n=-K^T\psi_{0}-\sum_{i=1}^{n-1}\psi_{i}/n$, we thus obtain:
\begin{equation}
\label{eq:regdens3}
\min_{\psi_{0},\cdots\psi_{n-1}} \sum_{i=1}^{n-1} H_{\nu_i}(\psi_{i}-K^T\psi_{0}/n)+H_{\nu_n}(-K^T\psi_{0}-\sum_{i=1}^{n-1}\psi_{i}/n)+E^*_\gamma(\psi_{0}).
\end{equation}
The subgradient \eqref{eq:subgrad} can then be used in a descent algorithm over the dual problem \eqref{eq:regdens3}.
For differentiable penalizers $E$, we consider the  L-BFGS algorithm \cite{Zhu:1997,lbfgs_implem} that integrates a line search method (see e.g.\ \cite{boyd2004convex}) to select the best time step $\tau^{(\ell)}$ at each iteration $\ell$ of  the subgradient descent:
\begin{equation}\label{alg:algo_2d}\left\{\begin{array}{lll}
\psi_0^{(\ell+1)}&= \psi_0^{(\ell)}-\tau^{(\ell)} (\nabla E^*_\gamma(\psi_0^{(\ell)})+d_0^\ell)\\
 \psi_i^{(\ell+1)}&= \psi_i^{(\ell)}-\tau^{(\ell)} d_i^\ell&i=1\cdots n-1,\\
 \end{array}\right.\end{equation}
where:{\small
\begin{equation*}\begin{array}{ll}
d^\ell_0&=K\left(\partial H_{\nu_n}\left(-K^T\psi_0^{(\ell)}/n-\sum_{i=1}^{n-1} \psi_i^{(\ell)}\right)-\sum_{i=1}^{n-1}\partial H_{\nu_i}\left(\psi_i^{(\ell)}-K^T\psi_0^{(\ell)}/n\right)\right)\\
d_i^\ell&=\partial H_{\nu_i}\left(\psi_i^{(\ell)}-K^T\psi_0^{(\ell)}/n\right)-\partial H_{\nu_n}\left(-K^T\psi_0^{(\ell)}/n-\sum_{i=1}^{n-1} \psi_i^{(\ell)}\right).
 \end{array}
 \end{equation*}
 }
 
 \noindent
  The barycenter is finally given by \eqref{eq:bar_reconstr}, taking $\phi_i=\psi_{i}-K^T \psi_{0}/n$. 
Even if we only treated differentiable functions $E$ in the theoretical part of this paper, we can numerically consider non differentiable penalizers $E$, such as Total Variation $(K=\nabla$,  $E=|.|_1$). In this case, we make use of the Fista algorithm. This  just modifies the update of $\psi_{0}$ in \eqref{alg:algo_2d}, by changing the explicit scheme involving $\nabla E^*_\gamma$ onto  an implicit one through the proximity operator of $E^*_\gamma$:
{\small\begin{equation*}
\begin{split}
\psi_0^{(\ell+1)}&=\textrm{Prox}_{\tau^{(\ell)} E^*_\gamma}\left( \psi_0^{(\ell)}-\tau^{(\ell)} d_0^\ell\right)=\uargmin{\psi}\frac1{2\tau^{(\ell)}}||\psi_0^{(\ell)}-\tau^{(\ell)} d_0^\ell-\psi ||^2+E^*_\gamma(\psi).
\end{split}
\end{equation*}}
 
\paragraph{Algorithmic issues and stabilization} As detailed in \cite{carlier2014numerical}, the computation of one subgradient in \eqref{eq:subgrad} relies on the look for Euclidean nearest neighbors between vectors $(y^{j}_{i},0)$ and $(x^k,\sqrt{c-\phi_i^k})$, with $c=\max_k\phi_i^k$. 
Selecting only one nearest neighbor leads to bad numerical results in practice as subgradient descent may not be stable. For this reason, we considered the $K=10$ nearest neighbors for each $j$ to build the row stochastic matrices $S_i$ at each iteration as:
$S_i^{j,k}={w_i^{jk}}/{\sum_{k'}w_i^{jk'}}$, with $w_i^{jk}=\exp(-(\frac12 \Vert y^{j}_{i}-x^k\Vert^2-\phi_i^k)/\varepsilon)$ if $k$ is within the $K$ nearest neighbors for $j$ and data $i$ and $w_i^{jk}=0$ otherwise.

\end{document}